\documentclass[runningheads]{llncs}

\usepackage{amsmath}
\usepackage{amssymb}
\usepackage{epsf}
\usepackage{graphics}
\usepackage{epsfig}
\usepackage{latexsym}
\usepackage{enumerate}
\usepackage[normalem]{ulem}
\usepackage{color}
\usepackage{float}
\usepackage{graphicx}
\usepackage[loose]{subfigure}
\usepackage{cite}
\usepackage{algorithm}
\usepackage{algorithmic}
\usepackage{epstopdf}
\usepackage{appendix}
\usepackage{placeins}

    \usepackage{changes}

\newcommand{\edge}[2]{(#1,#2)}
\newcommand{\seppair}[2]{\langle#1,#2\rangle}

\begin{document}

\title{Embedding 1-Planar Graphs in Ten Pages
\thanks{Supported by the Deutsche Forschungsgemeinschaft (DFG), grant
    Br835/20-1.}}
\author{Franz J. Brandenburg}
\institute{94030 Passau, Germany \\
           \email{brandenb@informatik.uni-passau.de}
}

\maketitle

\begin{abstract}
Every planar graph has a 4-page book
embedding  and  this bound is tight.
We show   that every 1-planar graph, which is a graph that
admits a drawing with at most one
crossing per edge,  has a 10-page book embedding. In addition,
four pages are sometimes necessary and always sufficient if the planar skeleton, obtained from a 1-planar drawing by removing all crossed edges, has a Hamiltonian cycle.
\end{abstract}
\vspace{2mm}


\section{Introduction} \label{sec:intro}
A  $k$-page \emph{book embedding}  of a graph consists of a
\emph{linear ordering} or \emph{layout} of the vertices, which is defined by placing
them from left to right, and an \emph{embedding}  of the edges in
$k$ pages, such that there is no  conflict  in any
page. For two vertices $u$ and $v$, let $u<v$ if $u$ precedes $v$ in
the linear ordering and let $u\leq v$ if $u<v$ or $u=v$.
Edges $\edge{u}{v}$ and $\edge{x}{y}$  with $u<v$ and $x<y$
 \emph{twist}  if $u < x < v < y$  or $x<u<y<v$.
 They \emph{nest} if $u \leq x < y \leq v$ and are
\emph{disjoint} if $u < v \leq x < y$. There is a \emph{conflict}
 if two   edges twist that  are embedded in the same page. This shall be excluded
whenever an edge is embedded in a page.

The \emph{book thickness} of a graph $G$ is the minimum number of
pages in all  book embeddings of $G$.  Book thickness is also known as
  stacknumber or pagenumber \cite{dw-llg-04,hls-cqsmlg-92}.
It has been  been shown that every nondiscrete  outerplanar graph has
 book thickness exactly one. A graph has book thickness at most two
  if and only if it is a subgraph of a planar graph with a Hamiltonian
cycle \cite{bk-btg-79}.  Every  planar graph  has book thickness at
most four. The upper bound has been shown by Yannakakis
\cite{y-epg4p-89} using a linear time algorithm that
constructs a 4-page book embedding. On the other hand,
  Bekos et al.~\cite{kbkpru-4pages-20} and  Yannakakis \cite{y-4pages-20}
  have shown that some planar graphs need four pages, so that the bound is tight.
The book thickness
of $n$-vertex graphs with $m$ edges is at most $\sqrt{m}$
\cite{m-edgepage-94} and at least $\lceil \frac{m-n}{n-3} \rceil$
\cite{bk-btg-79}. In particular, the complete graph $K_n$ has book
thickness $\lceil n/2 \rceil$ \cite{bk-btg-79}.
 Moreover, every graph has a subdivision that can be
embedded in three pages \cite{bk-btg-79}.

Alternatively, there are queue layouts which admit
twists and disjoint edges but exclude nesting
\cite{dw-llg-04, hls-cqsmlg-92, hr-loguq-92}. It is known that graphs
with   1-queue layouts are planar \cite{hr-loguq-92} and that the planar graphs
have bounded queue number \cite{djmmuw-queue-20}.
The planar graphs have be characterized by
an embedding in a
splittable deque  \cite{abbbg-deque-18}, which
is an advanced data structure consisting of a doubly connected list
(or deque \cite{k-acp-68}) on top of a depth-first search tree.

There are several approaches to extend the planar graphs, for
example by drawings on surfaces of higher genus \cite{gt-tgt-87},
forbidden minors \cite{d-gt-00},   drawings in the plane with
restrictions on crossings \cite{dlm-survey-beyond-19}, or
generalized adjacency relations \cite{cgp-mg-02}.
   Graphs with bounded genus have constant book thickness \cite{m-genuspage-94}.
Also minor-closed graphs, e.g., graphs with constant tree-width,
have constant book thickness \cite{dw-gtgtp-07}.
 A graph is $(g, k)$-\emph{planar}  if it can be drawn on a surface of Euler genus at
most $g$ with at most $k$ crossings per edge
  \cite{df-stackqueue-18}. Clearly,  planar graphs are $(0, 0)$-planar
  and $(0, k$)-planar graphs are known as
$k$-\emph{planar graphs} \cite{pt-gdfce-97}. An $n$-vertex
$(g,k)$-planar graph with fixed $g$ and $k$ has book thickness
$O(\log n)$  
\cite{df-stackqueue-18}.
 For $k$-planar graphs, this improves the $O(\sqrt{n})$ bound from
\cite{m-edgepage-94}.
 In a  manuscript,   Alam et al.~\cite{abk-bt1pg-15} have proposed
an algorithm that embeds   3-connected   1-planar graphs  in twelve
and general ones in sixteen pages, but there are some gaps.  Bekos~et
al.~\cite{bbkr-book1p-17} have shown that the book thickness of 1-planar
graphs is constant and that 39 pages suffice.
 From a recent approach
by Bekos et al.~\cite{bdggmr-framed-24}
one can obtain that the book thickness of 1-planar graphs is at most 31, which
has been improved to eleven in  \cite{b-bookmap-20}.  We save one more page by an
approach  designed for 1-planar graphs.

For an upper bound on the  book thickness, we   augment a given 1-planar
graph by further edges and
 allow multi-edges. So we obtain
  \emph{normalized 1-planar multigraphs} which have a 1-planar drawing
with triangles and X-quadrangles. An X-\emph{quadrangle}
is a quadrangle  with uncrossed edges
in the  boundary and a pair of crossed edges in the interior.
A  planar multigraph with triangular and quadrangular faces remains
if all crossed edges are removed and there is a
triangulated planar multigraph if only one edge from
each pair of crossed edges is removed.
Such graphs are 2-connected but not necessarily 3-connected.
Our algorithm uses the  peeling technique  that
iteratively decomposes a graph into 2-level graphs.
 Levels are defined
inductively on a planar multigraph with the vertices in the outer face at level zero.
 A planar 2-\emph{level multigraph} consists of an outer cycle
at level $\ell$ and the subgraph in its interior excluding the
subgraphs inside cycles at level $\ell+1$.  The subgraph in the interior
is composed of \emph{blocks} which are the cycles of 2-connected components.
Blocks are organized in a \emph{block-tree}  which is a connected
component of blocks.

 Yannakakis~\cite{y-epg4p-89} has proposed two methods
 for the computation of the vertex ordering. In the \emph{consecutive method},
each block is traversed individually, whereas
all blocks with the same dominator are traversed in one sweep
in the \emph{nested method}. For any
method, he has shown that every planar 2-level  graph can be embedded in
three pages  if the graph is internally triangulated and  there is a single block-tree.
These restrictions can be relaxed.
We introduce \emph{block-expansions} as a new method. It   extends
the nested method and defines  super-blocks that take on the role of blocks
from other approaches.
Thereby, we obtain a  new vertex ordering over other works  \cite{bbkr-book1p-17, bdggmr-framed-24,y-epg4p-89}.
In general,  a 3-page book embedding for planar 2-level graphs is used.
Instead,
we aim at a 4-page book embedding
for 1-planar 2-level  graphs that includes
almost all edges. Only the purple ones are left,
that are defined below and need two extra pages.
So we obtain  a 6-page embedding for any  2-level subgraph of a
normalized 1-planar multigraph.

The peeling technique on planar multigraphs allows to reuse pages at
 all odd (even) levels, so that the total number of pages is at
most twice the number of the 2-level case.
There is an improvement  at the composition of the
embedding  of 2-level graphs.
Yannakakis 5-page algorithm  saves one page for the uncrossed edges.
Similarly, one page can be saved for the purple ones,
such that ten pages suffice.
  In comparison, Bekos et al.~\cite{bbkr-book1p-17} have
used 16 pages for 1-planar 2-level graphs, 34 pages for the
composition of book embeddings from 2-level graphs and 5 additional
pages to deal with   separation pairs and
edges that cross in the outer face.

Alam et al.~\cite{abk-bt1pg-15} have shown that four pages suffice if the
 1-planar graph contains a Hamiltonian cycle when all crossed edges are removed.
This holds, e.g., for extended wheel graphs, which are  $n$-vertex   1-planar
graphs with the maximum of $4n-8$ edges
\cite{b-ro1plt-18,s-s1pg-86,s-rm1pg-10}. On the other hand,
such 1-planar graphs  need  at least four pages. Hence,
extended wheel graphs have book thickness four, the smallest of which
has eight vertices. \\

\noindent \textbf{Our contribution.}
We first
 extend  previous approaches towards a simplification of 1-planar graphs
and augment a 1-planar drawing
by uncrossed multi-edges, such that there is a normalized 1-planar multigraph.
Next, we introduce  block-expansions  as a new method for the computation of the
vertex ordering and use a novel page assignment for the edges that embeds almost
all edges of a 1-planar 2-level graph in four pages. The remaining edges
can be embedded in two more pages. Two pages can be saved
at the composition, so that we obtain a 10-page book embedding
for any   1-planar graph.
Finally, we provide simple examples
of 1-planar graphs with book thickness exactly four.\\

After   introducing   basic notions in  the next Section, we    describe an algorithm for a
6-page book embedding of 2-level  graphs
 in Section~\ref{sec:2-level}, and then   combine it to a 10-page book
embedding of   1-planar graphs in Section~\ref{sect:main}. We
conclude in Section~\ref{sec:conclusion} with some open problems.

\section{Preliminaries}  \label{sec:prelim}
We consider  undirected \emph{graphs} and \emph{multigraphs} $G = (V,E)$ with sets of
 vertices $V$ and edges $E$ without self-loops. Some edges may have copies in a planar or 1-planar drawing,
such that there are parallel edges. 
  A set of edges with a  common vertex
$v$ is called a \emph{fan} at $v$. For a subset of vertices $V' \subseteq V$
let $G[V']$ denote the subgraph induced by $V'$. Similarly, let
 $G[E'] = (V,E')$ for  $E' \subseteq E$.
 The set of vertices of an induced subgraph $G'$
is denoted by $V(G')$.

We assume that a    graph is simple and is defined by a drawing
 (or a topological embedding) in the plane.
A drawing is \emph{planar} if edges do not cross and \emph{1-planar}
if each  edge is crossed at most once.  A planar drawing  partitions
the plane into \emph{faces}.
 The unbounded face is called the \emph{outer face}.
 For convenience,
we will use (mixed) sets of vertices and edges
to specify a face, e.g., a vertex $a$ and an edge $\edge{b}{c}$ for
a triangle $(a,b,c)$, and say that a vertex (edge) is \emph{in the face}
if it is in the boundary.
Ringel~\cite{ringel-65} has first observed that
uncrossed edges $\edge{a}{b}, \edge{b}{c}, \edge{c}{d}$ and $\edge{d}{a}$
can be added if  edges $\edge{a}{c}$  $\edge{b}{d}$ cross in a 1-planar drawing,   see Figure~\ref{fig:framedWconf}. Added edges may be a copy,
such that the obtained graph is augmented to a multigraph.
Thereafter, edges cross only in the interior of a quadrangle,
such that there is an \emph{X-quadrangle}, also known as an
X-configuration~\cite{abk-sld3c-13} or  a kite~\cite{bdeklm-IC-16}, see
Figures~\ref{fig:kite}-\ref{fig:framedWconf}.
Hence, we can assume that there is no crossing in the outer face
of a 1-planar drawing, since we can use a multi-edge.
 As usual, a  graph is \emph{planar} (1-\emph{planar})
if it admits a  planar (1-planar) drawing.
It is well-known that every
   $n$-vertex 1-planar graph has
at most $4n-8$ edges, and that the bound is tight for $n=8$
and all $n \geq 10$  \cite{bsw-1og-84}.
For convenience, we use terms like vertex, edge, face or crossing
both for a planar (1-planar) graph and its planar (1-planar) drawing.

\begin{figure}[t]
  \centering
\subfigure[]{
    \includegraphics[scale=0.3]{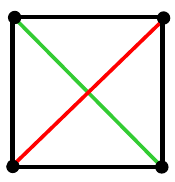}    
    \label{fig:kite}
    }
  \subfigure[]{
    \includegraphics[scale=0.32]{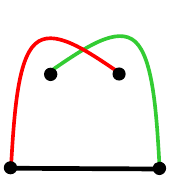}   
    \label{fig:Bconf}
    }
 \subfigure[]{
    \includegraphics[scale=0.32]{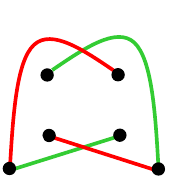}    
    \label{fig:Wconf}
    }
 \subfigure[]{
    \includegraphics[scale=0.32]{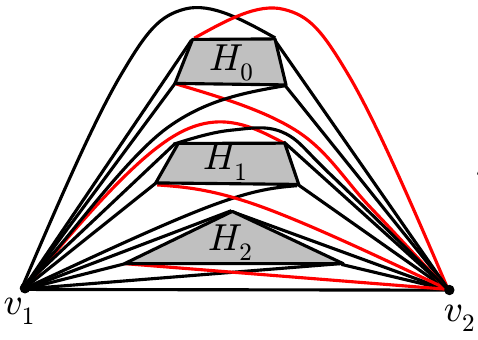}    
    \label{fig:3components}
    }
 \subfigure[]{
    \includegraphics[scale=0.4]{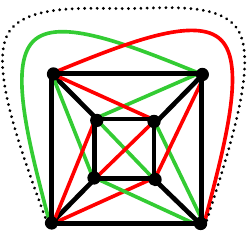}  
    \label{fig:crossed-cube}
    }
 \subfigure[]{
    \includegraphics[scale=0.28]{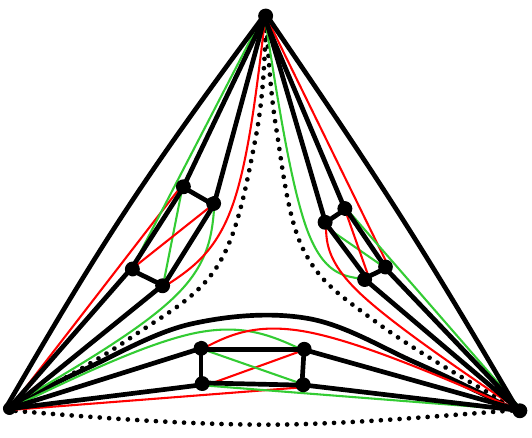}  
    \label{fig:framedWconf}
    }
\caption{(a) An X-quadrangle, (b) a B-configuration, (c) a W-configuration, (d)
a separation pair with B- and W-configurations,
(e) the 1-planar crossed cube with a multi-edge drawn dotted, and (f) a normalized 1-planar multigraph.
  }
  \label{fig:1planar}
\end{figure}

A 1-planar   drawing may contain B- and W-configurations \cite{t-rdg-88}, as shown
in Figures~\ref{fig:Bconf} and \ref{fig:Wconf},
that have a crossing  in the outer face of a 1-planar drawing of the graph or
a component at a separation pair \cite{abk-sld3c-13}.
We wrap these configurations by a multi-edge.
A 1-planar   drawing is \emph{normalized} if it consists of triangles
and X-quadrangles, possibly  with multi-edges, such  that
every copy of a multi-edge is uncrossed.
The edge itself is taken as 0-th copy,
and   there are
vertices on either side of a cycle with   two copies,
except for the outer face.
A planar drawing is normalized if it is obtained from a
normalized 1-planar drawing by removing one or both edges from each pair
of crossed edges.
A 1-planar (planar) multigraph is \emph{normalized} if it admits a respective
drawing. A normalized 1-planar (planar) multigraph is 2-connected but not
necessarily 3-connected, see  Figures~\ref{fig:3components} and \ref{fig:framedWconf}.
A  \emph{separation pair}  $\seppair{u}{v}$
consists of two vertices $u$ and $v$, such that  $G-\{u,v\}$
partitions into at least two connected components,
one of which is the \emph{outer component}, that is, it contains vertices
in the outer face of the given  drawing of $G$. The other components
are called \emph{inner components}.
It  is an \emph{inner} (\emph{outer}) separation pair  if $u$ is not
  in the outer face (if $u$ and $v$ are in the outer face). \\

 Similar to \cite{abk-sld3c-13, bbkr-book1p-17, ehklss-tm1pg-13b},
we augment a  1-planar drawing  by uncrossed multi-edges
towards a triangulation.

\begin{lemma} \label{lem:normalform}  
Every  1-planar graph $G=(V,E) $ can be augmented to
 a normalized 1-planar multigraph $G'=(V',E')$
with $V=V'$ and  $E \subseteq E'$, such that  $G'$ has an uncrossed  multi-edge
\edge{u}{v}
  if   there is a separation pair $\seppair{u}{v}$.
Conversely, there is a separation pair or the outer face (in a 1-planar drawing) if there is an uncrossed  multi-edge.
Graph $G'$ has at most $4n-7$ edges, including copies, if it has $n$ vertices.
\end{lemma}

\begin{proof}
If  a face of a 1-planar drawing of the simple graph $G$
contains both vertices of a crossed edge, then reroute the edge in the interior of the face.
Thereby one crossing is removed.
Thereafter,
if edges $\edge{a}{c}$ and $\edge{b}{d}$ cross, then  the drawing is augmented by
uncrossed  edges such that there is an X-quadrangle, as
 suggested by Ringel \cite{ringel-65}. An added  edge  may be a copy,
that is omitted  if there is no
vertex in the interior of the 2-cycle consisting of any two copies.
Finally, the obtained drawing is triangulated,  such that there are no further
crossings or multi-edges. Therefore, observe that
 any face $f$ with vertices $v_1,\ldots, v_k$
for   $k \geq 4$ has a vertex $v_i$ such that $f$ can be partitioned
 by a  chord $\edge{v_i}{v_j}$ for   $j \not\in \{i-1, i, i+1\}$,
such that uncrossed and  not a copy. Clearly, $\edge{v_i}{v_j}$ is uncrossed if
it is drawn in the face. It is not a copy  for $k=4$,
since, otherwise, edges $\edge{v_1}{v_3}$ and $\edge{v_2}{v_4}$ must cross outside $f$,
which is excluded by the first step, and, by induction, for any $k \geq 5$.
Obviously, the so obtained graph $G'$
is a normalized 1-planar multigraph.

If $\seppair{u}{v}$ is a separation pair of $G'$, then
there is a face in the drawing of $G-\{u,v\}$ between any two components.
Hence, an uncrossed   copy of $\edge{u}{v}$ can be drawn in the face
between two components that are consecutive at $u$.
Conversely, the curve formed by two uncrossed copies of an edge partitions
(the drawing of) $G$ into an inner and an outer part if each part is non-empty.
Otherwise, one copy is omitted or the graph is completely in the interior.

At last, consider the number of edges of $G'$ including copies.
Every simple  $n$-vertex 1-planar graph (without multi-edges) has at most $4n-8$ edges
\cite{ringel-65}. Hence, $G'$ has at most $4n-8$ edges if it has no separation pairs.
By induction on the number of separation pairs, if $G-\{u,v\}$ decomposes
into parts $G_1$ and $G_2$ with $k$ and $n-2-k$ vertices, respectively,
then $G_1+\{u,v\}$ has at most $4(k+2)-8$
edges and $G_2+\{u,v\}$ has at most $4(n-k)-8$ edges including copies
if edge $\edge{u}{v}$ is
discarded for both parts. Now the number of edges of $G'$ is
at most $4(k+2)-8 + 4(n-k)-8+1 = 4n-7$, since  a copy of
edge $\edge{u}{v}$ is added in the end.
\end{proof}

Note that  the bound of $4n-7$ is tight, as shown by the crossed
cube in Figure~\ref{fig:crossed-cube}. In fact, any simple $n$-vertex 1-planar
graph with $4n-8$ edges has two vertices in the outer face that can be connected
by an edge and its first copy.
Moreover, any two   components at a separation pair
in a normalized 1-planar drawing are separated by an uncrossed copy
of a multi-edge that excludes a crossed
edge between any two vertices of the components, see
Figures~\ref{fig:3components} and \ref{fig:framedWconf}.
In consequence, the  book embedding of the inner components at a separation pair
of a 1-planar graph will be
simpler than the one for 3-connected graphs,
as opposed to  \cite{bbkr-book1p-17}, where extra pages are used.

Some \emph{coloring schemes}  for the edges of a 1-planar (multi-)
graph have been proposed \cite{ehklss-tm1pg-13b, el-racg1p-13}.
Here the uncrossed edges in a 1-planar drawing are first colored
\emph{black}  and  the crossed ones   \emph{green},
\emph{red} or \emph{purple}. The color is specified later, where
  we will recolor  some black binding edges, called pre-cluster edges.
If $E_b$  is the set of black   edges, then $G[E_b]$
is the \emph{planar skeleton} of $G$. Here all crossed edges are
removed.  We obtain a   triangulated
planar multigraph, that may have separation pairs, if exactly one edge
from each pair of crossed edges is removed.

At last, we need some terms for a linear ordering of the vertices. For sets of
vertices $U$ and $W$ let $U<W$ if $u<w$ for every $u \in U$ and $w
\in W$.
 An edge $\edge{u}{w}$ \emph{spans} vertex $v$ if $u < v <
w$. Thus two edges with
distinct vertices twist if and only if each of them
 spans exactly one vertex of the other edge.
An \emph{interval} $[u,w]$
consists of all vertices $v$ with $u \leq v \leq w$. Vertex $v$  is \emph{inside}
$[u,w]$ if $u \leq v \leq w$ and
\emph{outside}  if  $v \leq u$ or $v \geq w$. Thus the
vertices on the boundary are both  in and outside. Obviously,
two edges  do not twist if there is an interval such that both vertices of one
of them are in  and the vertices of the other   are outside the
interval. Let $[W]$  denote the interval containing exactly
 the vertices of a set $W$ and let $[W,v] = [w,v]$,
where $w$ is the
least vertex of $W$, and similarly for $[v,W]$.

\section{Book Embedding  of 2-Level Graphs} \label{sec:2-level}

 The \emph{peeling technique}, introduced by Heath~\cite{heath1984embedding},
has been used in all  later approaches for an upper bound of the book
thickness of generalized planar graphs \cite{ 
bbkr-book1p-17,bdggmr-framed-24, df-stackqueue-18}, and, notably, in
Yannakakis 4-page algorithm for the planar graphs \cite{y-epg4p-89}.
It  decomposes a graph into 2-level graphs and computes a \emph{leveling}  of the
vertices of a graph, such that there are layered separators
\cite{d-glls-15}.

We use the planar skeleton $G[E_b]$
of a normalized 1-planar multigraph $G$.
The vertices in the outer face of (a planar drawing of) $G[E_b]$ are at level
zero.  Vertices are at level $\ell+1$ if they are in the outer face
when all vertices at levels at most $\ell$
are removed. Obviously, there are no edges between vertices in levels $i$ and
$j$ for  $|i-j| >1$ in $G[E_b]$. By the normalization, this also holds for $G$. In consequence,
the book embedding of a 1-planar graph can be composed of the book embedding
of its 2-level subgraphs at odd and even levels, so that the book
thickness  of $G$  is at most twice the book thickness of its 2-level subgraphs.

In the remainder of this Section, we first compute the vertex ordering
using the planar skeleton $G[E_b]$.
We will extend Yannakakis nested method \cite{y-epg4p-89}
 and combine  distinguished sets  of blocks into   a cluster  and a
 super-block, respectively.  Then we color edges and
obtain  an embedding of any  2-level 1-planar (multi-) graph  in six pages.

\subsection{Planar 2-Level Graphs} \label{sect:vertex-ordering}
We recall basic notions from \cite{y-epg4p-89} and extend them for our approach.
Familiarity with Yannakakis algorithm on planar 2-level  graphs will be helpful.

Forthcoming, let $H$ be a
normalized planar multigraph  induced by
a cycle $O=v_0,\ldots, v_t$ of level $\ell$ vertices, called \emph{outer vertices},
and of the level $\ell+1$ vertices in the interior of $O$, called
\emph{inner vertices}, where $\ell \geq 0$. Graph $H$ is a 2-level
subgraph of the planar skeleton $G[E_b]$, where $G$ is a normalized 1-planar multigraph.
The   vertices of $O$  are ordered $v_0 < \ldots < v_t$, which is the clockwise traversal
(cw-order) of $O$ from $v_0$.   We have $t \geq 1$ due to multi-edges.
The inner vertices induce the \emph{inner subgraph}  that is composed of 2-connected components.
The outer cycle of each 2-connected component is called a \emph{block},
 see Figure~\ref{fig:levelgraph}.
A block may be small and consist of a single vertex or of two vertices with
a (multi-) edge.
Blocks are traversed in opposite direction of $O$, that is in
counterclockwise order (ccw-order).
The edges of $H$ are \emph{outer edges} in the
outer cycle $O$, \emph{outer chords} between non-consecutive outer
vertices in the interior of $O$, \emph{binding edges} between inner
and outer vertices (in this direction), that are classified into
\emph{forward} and \emph{backward binding}, and \emph{inner edges}
between two vertices that are consecutive for a block.
Outer chords outside $O$ have been treated at the previous level. An
 \emph{inner chord}  between two non-consecutive inner vertices of a block
is either a crossed edge of $G$ or  it shall be considered at the next level.

\begin{figure}[t]
\centering
\includegraphics[width=0.7\textwidth]{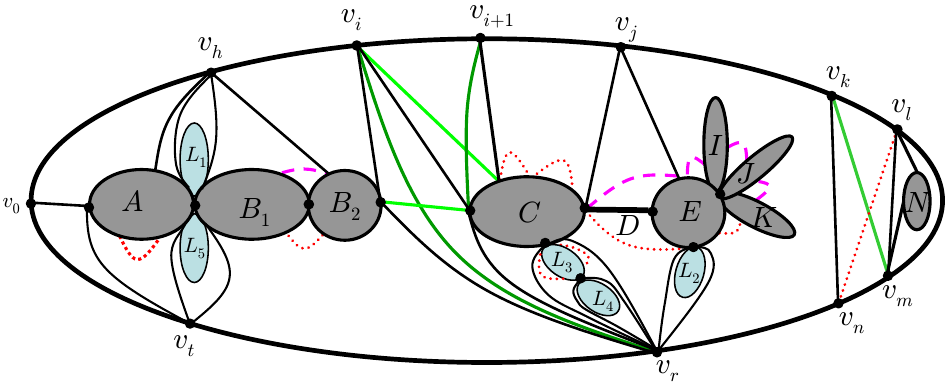} 
  \caption{ 
A sketch of a 1-planar 2-level graph with blocks drawn gray if they are uncovered
and blue if they are   covered.~Red edges are dotted and purple ones
are dashed.~Most vertices and binding edges
are omitted to avoid an overloading  of the picture.~The vertex ordering is
  $v_0,  \{A,  L_5, L_1\}, v_h,   B_1,B_2, v_i,
  v_{i+1},  \{C, L_3, L_4\}, v_j, D,\{E,L_2\} , \ldots,v_l, N, v_m, \ldots, v_t$,
where super-blocks are in brackets.
  }
  \label{fig:levelgraph}
\end{figure}

A face  between the outer cycle
  and the inner subgraph is called an \emph{interior face}.
It is a triangle  or a quadrangle, where each quadrangle contains a pair
of crossed edges of $G$. 
Each interior face $f$ has
outer vertices and thus a \emph{first outer vertex} $\alpha(f)$  and
a \emph{last outer vertex} $\omega(f)$, which are the least and the last
outer vertex in the boundary of $f$. Clearly, $\alpha(f)=\omega(f)$
is possible, and all (three or four) vertices of $f$ can be outer vertices.

The set of blocks   $M$ of the inner subgraph is structured as follows.
 Any two blocks may
share a vertex, which is a \emph{cutvertex} of $M$.
A connected component of
$M$ is called a \emph{block-tree}, which is a cactus consisting of
blocks with branches at cutvertices. Each block-tree has a distinguished first block, called the \emph{root}.
Any two block-trees are separated by an outer  chord   or by a
quadrangle  that can contain an outer chord.
Each block $B$ has a least vertex $\lambda(B)$, called the
\emph{leader}, which is the cutvertex of $B$ and its parent
if $B$ is not the root of a block-tree.  A cutvertex may be the
leader of several blocks,  that are \emph{siblings}  and are
ordered clockwise at the cutvertex. The \emph{root}  of a block-tree is special.
It may be \emph{elementary} and consist of a single vertex, see
Figure~\ref{fig:elementary}. In fact, the inner subgraph may be empty or consist
of a single vertex.
A block-tree $\mathcal{T}$ has a \emph{first face}
$f_{\mathcal{T}}$, which is the least face containing a vertex of any block of
$\mathcal{T}$.  
In addition, it has a last binding edge $\edge{b_0}{v}$ in its boundary,
whose vertices   $b_0$ and $v$ are called
the   \emph{first vertex} and the \emph{last outer vertex}
 of $\mathcal{T}$,  denoted $\lambda(\mathcal{T})$
and $\omega(\mathcal{T})$, respectively.
Vertex $b_0$ is the leader of the root of $\mathcal{T}$.
The first outer vertex of $f_{\mathcal{T}}$ is denoted by $\alpha(\mathcal{T})$,
and is called the \emph{first outer vertex}    of $\mathcal{T}$.
Thus $f_{\mathcal{T}}$ contains the vertices
$\alpha(\mathcal{T}), \lambda(\mathcal{T})$ and $ \omega(\mathcal{T})$
and  another inner or outer vertex if it is a quadrangle. 

 Each inner vertex is in a single block, except if it is a cutvertex. For
  uniqueness,   the leader is assigned to the block that is closest to the root.
This   reflects the ordering of the vertices \cite{y-epg4p-89}.
We denote the \emph{set of vertices assigned to block} $B$ by $V(B)$.
If $B$ is a non-root block with vertices $b_0, b_1,\ldots, b_q$,
then $b_0 \not\in V(B)$, $b_0$  is the leader of $B$, that is $b_0=\lambda(B)$,
$b_1$ and $b_q$ are the
\emph{first}   \emph{last}  vertex and $\edge{b_0}{b_1}$
and   $\edge{b_0}{b_q}$ \emph{first} and   \emph{last   edge} of $B$, respectively.
This is assumed from now on.
Moreover, we say that   vertex $b$
 is in  block $B$ and that $B$ is the block of $b$ if it is clear  from the context,
and similarly for edges and block-subtrees.

The \emph{dominator} of block $B$, denoted $\alpha(B)$, is the
least outer vertex in the face containing  the last edge of $B$.
Vertex $\alpha(B)$ \emph{sees} $B$ according to     \cite{y-epg4p-89}.
 Similarly, there is a \emph{last outer vertex}
$\omega(B)$, which is the  least outer vertex in the face containing
the first edge   of $B$. Hence, any block $B$ has three distinguished vertices
 $\alpha(B), \omega(B)$  and $\lambda(B)$, or $v_i, v_k$ and $b_0$
if this is appropriate.
Note that there may be edges $\edge{\lambda(B)}{v}$ with outer vertices $v$  such that
  $v < \alpha(B)$ and  $v > \omega(B)$, respectively.
Obviously, for any block $B$ of a block-tree $\mathcal{T}$, we have
$\alpha(\mathcal{T})\leq\alpha(B)$ and $\lambda(\mathcal{T})\leq\lambda(B)$,
where equality holds for the root,
and $\omega(\mathcal{T})\geq\omega(B)$, where $\omega(\mathcal{T})>\omega(B)$
for every block $B$ of $\mathcal{T}$ is possible.

The ordering of the vertices on the outer cycle induces an ordering of the interior faces,
blocks, block-trees, outer chords, and binding edges, denoted by
$x \prec y$ if $x$ precedes $y$. The interior faces are ordered clockwise
according to their first outer vertex and in counterclockwise  order  if faces have the same first outer vertex.
 If $x$ and $y$ are a block, a block-tree, (a copy of) an
  outer chord, or a (copy of a) binding edge,
 then let $x \prec y$ if the first face containing an element of $x$ precedes
the first face containing an element of $y$.
In particular,   blocks are ordered according to their dominator, as in \cite{y-epg4p-89},
and in ccw-order at a common dominator.

For our notation, we will use capital letters $A,B,C$ for blocks, small letters $a,b,c$
for inner vertices, $u,v$ for outer vertices,  and  $\alpha, \lambda$ for least
and $\omega$ for last elements. \\

\begin{figure}[t]  
\centering
\subfigure[ ] {    
    \includegraphics[scale=0.6]{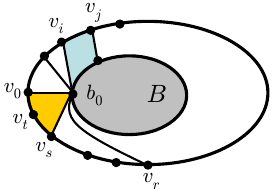}  
      \label{fig:elementary}
\hspace{2mm}
  }
\subfigure[ ] {    
    \includegraphics[scale=0.6]{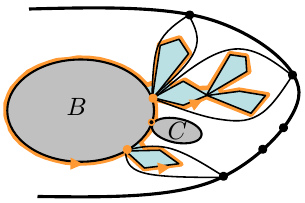}  
      \label{fig:super-block}
\hspace{2mm}
  }
\subfigure[ ] {    
    \includegraphics[scale=0.7]{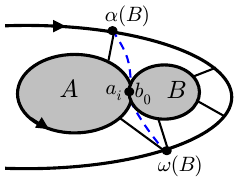}   
      \label{fig:separate}
  }
 \caption{(a) An elementary block $b_0$. (b) Block $B$ and its expansion to
a super-block $B^+$, and
(c) a partition   induced by block $B$, indicated by the blue dashed line.
  }
  \label{fig:firstface}
\end{figure}

 Yannakakis \cite{y-epg4p-89} has proposed two methods for the traversal of blocks.
Each block $B$ is traversed
individually in ccw-order in the \emph{consecutive method},
such that its vertices are consecutive at this moment,
except for the leader, which is in the parent block, in general.
Blocks   with the same dominator are visited one after another
in ccw-order at the dominator,
where the blocks may be in one or more block-trees.
  In the \emph{nested method},  the  set of blocks of  a block-tree
with the same dominator is traversed by depth-first search \cite{clrs-ia-01}.
Their boundary is traversed in a since sweep in ccw-order.
Each   vertex is listed exactly once at its first appearance.
 The nested method lays out the vertices of $B$ in a series of intervals
  between two consecutive  cutvertices of $B$,
such that the structure of the intervals one-to-one corresponds
to the traversed block-subtree.
In both methods, the visited vertices are
 placed  to  the right of the dominator, and to the left of its
successor if the dominator is the least outer vertex $v_0$.

Yannakakis \cite{y-epg4p-89} has proved that
  $\alpha(A) \leq \alpha(B)$ and  $\omega(B) \leq \omega(A)$
  if  block $A$ is the parent of $B$ and his requirements are fulfilled, in particular
$\alpha(B) < V(B) < \omega(B)$ for any block $B$.
   Hence, there is an interval $[\alpha(B), \omega(B)]$ within $[\alpha(A),
   \omega(A)]$ in the vertex ordering,
such that $[\alpha(B), \omega(B)]$ does not contain vertices of $A$
if $\alpha(A)<\alpha(B)$.
This relation  holds for both  methods. The leader $\lambda(B)$ resp.
$\lambda(A)$ is outside the interval, in general.

We need a generalization, since there are separation pairs. These
are unavoidable for  1-planar graphs, whereas they can be
excluded for planar
graphs \cite{y-epg4p-89}.  In fact, if $\edge{x}{y}$ is any uncrossed edge in a 1-planar
drawing of $G$, then $G$ can be extended by   inner components, such that
$\seppair{x}{y}$ is a  separation pair,
for example, by adding a vertex $z$ and edges $\edge{x}{z}$ and $\edge{y}{z}$
or a graph with a W-configuration.
We say that   block $B$   is \emph{covered} by the outer vertex $v$ if
$v=\alpha(B)=\omega(B)$,
see Figures~\ref{fig:levelgraph} and \ref{fig:super-block}.
Then all binding edges with an inner vertex in $B$ are incident to $v$.

 The following Lemma extends statements by Yannakakis \cite{y-epg4p-89}, see
also  \cite{bbkr-book1p-17}.

\begin{lemma}     \label{lem:partition}    
Let $H$ be a normalized planar 2-level multigraph.

\begin{enumerate}[(i)]
 \item
Each  block $B$ has a first and a last outer vertex
such that $\alpha(B)\leq \omega(B)$.
\item
Block $B$   is in an inner component at an inner separation pair
 if and only if  $\alpha(B)=\omega(B)$, that is $B$ is covered.
\item
If $B$ is covered by $v$,
then the face    containing
 $v$  and  the first and last  edge of $B$, respectively,
contains an uncrossed copy of edge $\edge{\lambda(B)}{v}$.
In other words,  $B$ is enclosed by an uncrossed multi-edge
$\edge{\lambda(B)}{v}$ with $v=\alpha(B)=\omega(B)$
and cutvertex $\lambda(B)$.
\item
For every uncovered block $B$,  graph $H-\{\alpha(B), \lambda(B), \omega(B)\}$
 partitions into parts $H_1$ and $H_2$, where
  $H_2$ contains the outer vertices
 $v_{i+1}, \ldots,  v_{k-1}$
 and the vertices of the  following blocks, where
$v_i = \alpha(B)$ and  $v_k=\omega(B)$:   block $B$,
  the blocks dominated by $v_i$ that succeed $B$ in
ccw-order at $v_i$,   the  blocks dominated
by $v_j$ for $i<j<k$,  and   all blocks $B'$  dominated by $v_k$, such
that any path from a vertex
of $B'$ to the root of the block-tree  
passes through or ends at  $\lambda(B)$. The other part is $H_1$.
\item
If $\edge{v_i}{v_k}$ is an outer chord or an outer edge, then
$H- \{v_i, v_k\}$ partitions similarly by its $j$-th copy.
Any block    dominated by $v_i$ is in $H_2$ if it
  succeeds the $j$-th copy of $\edge{v_i}{v_k}$,
and any block in $H_2$  that is
dominated by $v_k$ has an ancestor   that succeeds $\edge{v_i}{v_k}$.
In other words, these blocks are to the left of the $j$-th copy of $\edge{v_i}{v_k}$,
that is $\edge{v_i}{v_k} \prec B$ for each such block $B$.
 The other part is $H_1$.
\end{enumerate}
\end{lemma}

\begin{proof}
Every interior face of a planar drawing of $H$ has an outer vertex and thus a first and last  outer
vertex. The vertices $\alpha(B)$ and $\omega(B)$  of block $B$ are
in the face containing the last and the first  edge of $B$,
respectively, if $B$ is non-elementary. As the outer cycle and blocks are
traversed in opposite direction, we have  $\alpha(B) \leq \omega(B)$.
If $B$ is elementary, then it is the root of a block-tree $\mathcal{T}$.
 Now $\alpha(B)$ and $\omega(B)$
are taken from the first face of $\mathcal{T}$, where
$\alpha(B) < \omega(B)$ holds.

For (ii), if $\seppair{b}{v}$ is an inner separation pair with an inner vertex
 $b$ and an outer vertex $v$,
then $\alpha(B) = \omega(B) = v$ for all blocks from  inner
components   of $H-\{b,v\}$. Hence, all these blocks are covered by
$v$. Conversely,
$\seppair{\lambda(B)}{v}$ is an inner separation pair if block $B$ is covered by $v$,
since there are faces containing $v$ and the first and last edge of $B$, respectively.

There is an uncrossed multi-edge
$\edge{\lambda(B)}{\alpha(B)}$ enclosing $B$ by Lemma~\ref{lem:normalform},
 such that the face containing
$\lambda(B)$ and the first and last edge of $B$, respectively, contains a copy of
$\edge{\lambda(B)}{\alpha(B)}$ in its boundary,
which proves (iii).

For (iv), if
   $B$ is not the root of a  block-tree, then its leader
 $b_0$ is a cutvertex of the inner subgraph  that is
partitioned by the removal of $b_0$. Clearly, the outer cycle $O$
is partitioned by the removal of $v_i=\alpha(B)$ and
$v_k=\omega(B)$. There is a closed curve $\Gamma$ from $v_i$ via $b_0$ to
$v_k$   that partitions $H$, see Figure~\ref{fig:separate}. The curve first follows (any copy of) the binding edge
$\edge{b_0}{v_i}$ that can be added in the interior of a quadrangle
if it does not exist. It passes $b_0$ and then follows (any copy of)
$\edge{b_0}{v_k}$ that can be added as well. The curve is closed along the outer face.
Hence, the   removal of $\{v_i, b_0, v_k\}$  decomposes $H$, such that vertices $v_{i+1},\ldots, v_{k-1}$ and
block $B$ are on one side of $\Gamma$, that is in $H_2$. By induction, also all
blocks dominated by $v$ with $v_{i+1} \leq v \leq  v_{k-1}$ are in $H_2$.
The blocks dominated by $v_i$ and $v_k$ need a special attention.
If block $B'$ is dominated by $v_i$, then $B'$ and $B$ are on the same side of $\Gamma$
  if and only if $B'$ succeeds $B$ in ccw-order at $v_i$. Then $B'$ is in part $H_2$.
If block  $B''$ is dominated by $v_k$ and is on the same side of $\Gamma$ as
$B$, that is  $B''$ is in $H_2$, then $\omega(B'') \leq  v_k$, such that
$B''$ is covered by $v_k$. Then there is a path from any vertex of $B''$ to the root of its
block-subtree through  $\lambda(B)$. Conversely, $B''$ is on the
same side   of $\Gamma$  as $B$  if any path from a vertex of $B''$ to the root of the block-tree
passes through   $\lambda(B)$.
If $B$ is the root of a block-tree  $\mathcal{T}$, then $H$ partitions similarly by
the removal of $\{\alpha(B), \lambda(B), \omega(B)\}$, where part $H_1$ may be empty
and paths from $B''$ end at $\lambda(B)$.

Obviously, the removal of the vertices of an outer chord $\edge{v_i}{v_k}$
decomposes $H$.
There is a curve $\Gamma$ using the $j$-th copy of $\edge{v_i}{v_k}$ for a
partition of $H$  into parts $H_1$ and $H_2$. Then all blocks that are
  dominated by $v_i$ and succeed the $j$-th copy are in the interior of $\Gamma$,
such that they are in part $H_2$, and similarly for all block dominated by $v_k$
that succeed the $j$-th copy of $\edge{v_i}{v_k}$, that is they are to its left.
Any such block has an ancestor that is dominated by $v_i$
and succeeds $v_i$.
Hence, $H$ partitions as claimed.
\end{proof}

We say that a partition of $H$  as in case (iv) and (v)  is
\emph{induced} by  a block and an outer chord or an outer edge, respectively.

The set of blocks of a block-tree with the same dominator forms a path if there
are no covered blocks, as shown in \cite{y-epg4p-89}. This
generalizes as follows.

\begin{lemma}  \label{lem:backbone}
If $S$ is the set of blocks of a block-tree $\mathcal{T}$ dominated by
a single outer vertex $v$, then $S$ is a block-subtree of
$\mathcal{T}$ consisting  of a path of uncovered blocks
$B_1,\ldots, B_k$ for some $k \geq 1$ to the right, called the \emph{backbone},
such that any other block  of $S$ is covered by $v$ and is to the left
of the backbone.
\end{lemma}

\begin{proof}
We proceed by induction. Let $B_i$ be the rightmost descendant of $B_{i-1}$,
that is    dominated  by $v$. Thus $\lambda(B_i)$ is the
 least vertex of $B_{i-1}$ that is a  leader of a block dominated by $v$
and $B_i$ is the last block in cw-order at $\lambda(B_i)$ that is
dominated by $v$. Then $B_i$ is uniquely determined, such that, by induction,
there is backbone $B_1, \ldots, B_k$ of uncovered blocks.
Any block $B$ with $\lambda(B)$ in $B_{i}$ for $i=1,\ldots, k-1$ and
$\lambda(B) \geq \lambda(B_i)$ 
and dominated by $v$
is covered by  $v$, since $v$ cannot dominate $B_{i+1}$, otherwise,
as the edges in the faces with the first edge of $B$ and the last edge of $B_{i+1}$
cross, which contradict planarity. As $B_k$ is the last uncovered block
dominated by $v$,   also  blocks $B$ with $\lambda(B)$ in $B_k$
are covered by $v$.
\end{proof}

We wish a vertex ordering $\alpha(B)< V(B) <\omega(B)$
for every block, as in \cite{y-epg4p-89},
which is impossible for covered blocks.  Instead, we merge
covered blocks into a super-block by a block-expansion,
and so remove all covered blocks.

An  inner separation pair  $\seppair{b}{v}$
is  \emph{maximal} if
  $b$ is  in a block that is not covered by $v$.
We comprise   all inner components of
$\seppair{b}{v}$ into a \emph{cluster}, denoted $C(b,v)$.
Vertices $b$ is an articulation vertex if $v$ is removed,
and is called  the \emph{leader} of $C(b,v)$. A
cluster consists of a set of block-trees that are separated by copies
of the multi-edge $\edge{b}{v}$ and are connected at $b$.
Clearly, all  blocks of $C(b,v)$ are covered by $v$.

\begin{definition}  
Let $H$ be  a normalized planar 2-level multigraph.
A   \emph{super-block} $B^+$    consists of the vertices of
 an uncovered block  $B$, called the \emph{root} of $B^+$,
  and for every   vertex $b$
assigned to $B$, the  vertices from all clusters
  $C(b,v)$.
\end{definition}

We say that   $B$ is \emph{expanded} by (the vertices, the
blocks of) a  cluster $C(b,v)$,
 see Figure~\ref{fig:super-block}.
There may be several outer vertices $w_1<\ldots< w_r$
with a cluster $C(b,w_i)$ at $b$. Also,  $B$ can be expanded at several of its vertices
with the same or with different outer vertices for the clusters. At last,
there can be inner vertices in several blocks with a cluster $C(a,w_i)$ for
the same outer vertex $w_i$.
As block $B$ is the root of $B^+$, every super-block is uncovered.
Conversely, if $A$ and $B$ are uncovered blocks, then the super-blocks
$A^+$ and $B^+$ with root $A$ and $B$ are different, even if $A$ and
$B$ have the same dominator. Hence, there is a one-to-one correspondence
between uncovered blocks and super-blocks. In other words, we absorb
covered blocks by super-blocks, but we need blocks for the coloring and
embedding of the edges, and for the outer cycle at the next level.
A cluster $C(b,v)$ is a block-subtree and so is a super-block.
 Clearly, every block is in exactly one super-block,
such that the inner subgraph of $H$ can be restructured.
Now super-blocks take  the role of blocks in \cite{y-epg4p-89}.
For convenience, we
adopt terms from blocks, such as   leader, dominator,   last outer vertex,
 parent, descendant or sibling. 

We expand a  block at any of its vertices and create large super-blocks,
which is an   ``eager version\rq{}\rq{} for block-expansions.
In \cite{b-bookmap-20} we have used a ``lazy version\rq{}\rq{},
where an uncovered block $B=b_0,\ldots,b_q$ is only expanded by
the clusters $C(b_i,v)$
if $v=\omega(\mathcal{T})$ is the last outer vertex of the block-tree
containing  $B$.
 The ``lazy version\rq{}\rq{} will do as well, with an adaptation
for the classification of edges and their assignment to pages.
We use the ``eager version\rq{}\rq{} as it is well-suited for 1-planar graphs.
 %
Super-blocks can be enlarged even more by using the nested method,
such that all super-blocks with the same dominator are
merged into a single one.
This is disadvantageous in our case, as the saving of
page $\chi(B)$ in Theorem~\ref{thm:10-pages} fails.

\begin{lemma} \label{lem:no-span}   
If  an  uncovered block $B$ with $v=\omega(B)$,
  an  outer chord or an outer edge $\edge{u}{v}$  induces
a partition of $H$,
 then $v$ does not dominate any super-block from the second part $H_2$.
\end{lemma}

\begin{proof}
Part $H_2$ is in the interior of a curve $\Gamma$ through vertices
$\alpha(B), \lambda(B)$ and $\omega(B)$, as shown in (the proof of)
Lemma~\ref{lem:partition}. If $D$ is a block in $H_2$ that is
dominated by $\omega(B)$, then it is covered by $\omega(B)$,
since $\omega(B)$ is the last outer vertex incident to vertices of $H_2$.
  Then $D$ is in a cluster $C(a, \omega(B))$ with vertex
$a$ in some uncovered block $A$, and $A$ is a block of $H_2$, since it is
in the interior of $\Gamma$.
%
Now  block  $A$ is expanded by $C(a, \omega(B))$
that includes block $D$.
By induction, any block in part $H_2$ that is dominated by  $\omega(B)$
is merged into an uncovered super-block, such that  $\omega(B)$ does not dominate
super-blocks in $H_2$.

The case  for  an outer chord or an outer edge $\edge{u}{v}$ is similar.
\end{proof}

\subsection{Vertex Ordering} \label{sect:vertexordering}

The \emph{vertex ordering} (or layout) $L(H)$ of a normalized
planar 2-level multigraph $H$ is computed as follows.
First, compute all super-blocks and their dominator.
The vertices   of the outer cycle are already ordered
$v_0 < \ldots < v_t$. For every outer  vertex $v_i$ with $i=0,\ldots,t$,
visit the super-blocks
dominated by $v_i$ in counterclockwise order  at $v_i$ and
place the  vertices  assigned to each of  them right after $v_i$.
The vertices are placed
  immediately to the left of $v_1$ if the least outer
vertex $v_0$ is the dominator,
as suggested by Yannakakis \cite{y-epg4p-89},
such that they are to the right of vertices
from the previous level, including the dominator of the outer cycle,
in general. Now we have
$v < V(A^+) < V(B^+) < v\rq{}$ if $v$ dominates super-blocks $A^+$
and $B^+$ with $A^+ \prec B^+$
and $v\rq{}$ is the immediate successor of $v$ on the outer cycle.
Hence, super-blocks are processed by the consecutive method, such that we have
an interval $[B^+]$ for $B^+$ in $L(H)$.

For the later page assignment, each cluster must be visited in one sweep,
such that is appears like a block. The boundary of the blocks of a cluster
$C(b,v)$ is traversed in  ccw-order and each cutvertex is listed at its first appearance
as in a depth-first search or pre-order traversal \cite{clrs-ia-01}. Thereby,
we extend Yannakakis nested method, where   the traversed blocks
form a path.
If $w_1<\ldots <w_r$ are outer vertices and $C(b_i,w_1), \ldots,  C(b_i,w_r)$
are  the clusters at $b_i$, then insert their vertices in the order
$C(b_i,w_r), \ldots, C(b_i,w_1)$ just right of $b_i$
for $B=b_0,\ldots, b_r$.
The layout of a super-block $B^+$ is obtained by traversing its root
$B$ in ccw-order, placing the vertices of $V(B)$ just
right of the dominator, and inserting the  vertices of  the clusters
at each $b_i$ just right of $b_i$ for each $b_i \in V(B)$.
Then $[b_i, b_{i+1}]$ contains exactly the
intervals $[C(b_i,w_r)], \ldots, [C(b_i,w_1)]$, where $b_{q+1}$ is the vertex
succeeding $b_q$ in $L(H)$. In consequence, the last vertex $b_q$ of block $B$
is no  longer the last vertex of $B^+$, since the clusters at $b_q$
are placed to the right of $b_q$. A cluster $C(b,w)$ is encapsulated in the sense
that all edges with a vertex  in $C(b,w)$ have the other vertex in
$C(b,w) \cup \{b,w\}$. Thus a placement to the right of the block
does not matter, as we will see.

There is a special case of an outer  separation pair $\seppair{u}{v}$
or a multi-edge $\edge{u}{v}$ with $u<v$.
Each inner component is a block-tree, where block-trees are
separated by a copy of the multi-edge $\edge{u}{v}$.
Each inner component has uncovered blocks $B_1,\ldots, B_r$
with dominator $u$  forming a backbone by Lemma~\ref{lem:backbone},
and blocks covered by $u$ and $v$, respectively.
A path of super-blocks $B_1^+,\ldots,B^+_r$ remains after
block-expansions. In the later book embedding, we will need
pages $P_2$ and $P_3$ alternatingly for the inner edges of two
consecutive super-blocks.
In this particular case, one may merge all inner components
into a single super-block.
In any case, there is an interval just right of $u$ containing
exactly all vertices from the inner components at $\seppair{u}{v}$.
Each vertex from an  inner component  has binding edges incident
only to $u$ or $v$, where the ones incident to $u$ are backward binding
and those incident to $v$ are forward binding.

\begin{corollary} \label{cor:chord-no-left}   
Suppose there is a partition of $H$ induced by a  super-block $B^+$ with
$v=\omega(B)$ or an outer (edge)
 chord $\edge{u}{v}$, as described in Lemma~\ref{lem:partition}.
Then there are no vertices of part $H_2$ to the right of vertex $v$.
\end{corollary}
\begin{proof}
Vertex $v$ does not dominate any super-block
of part $H_2$ by Lemma~\ref{lem:no-span}. Hence, there are
no vertices of $H_2$ to the right of $v$ by our vertex ordering.
\end{proof}

Corollary~\ref{cor:chord-no-left} describes a benefit from the
use of super-blocks, namely, that the last outer vertex   of  a
super-block, an outer edge
or an outer chord is a strict right boundary  for the vertices in  part $H_2$ from a
partition of $H$. This will be important as several places, for example in
Lemma~\ref{lem:x1}.
It also holds for  the last outer vertex $\omega(\mathcal{T})$
of a  block tree $\mathcal{T}$, since $\omega(\mathcal{T}) \geq \omega(B)$
for any (super-) block of $\mathcal{T}$.
In particular, the last outer vertex $\omega(\mathcal{T})$
may dominate and thus cover blocks of  $\ \mathcal{T}$.
These blocks   are treated as
the   special case of a ``small face'' in
the approach by Bekos et al.~\cite{bdggmr-framed-24}.
Here they change the common rule
and  place the vertices from the covered blocks immediately to the left
of $\omega(\mathcal{T})$, which is problematic \cite{bdggmr-benpsf-20}
and costs two extra pages  \cite{bdggmr-framed-24}.

We use   the consecutive method for  super-blocks, the nested method
 for clusters, and block-expansions for super-blocks.
This combination  is superior to a single method.
A  super-block $B^+$  is
treated like a block in other approaches \cite{bbkr-book1p-17,bdggmr-framed-24, y-epg4p-89}, since  its  vertices form
 an interval $[V(B^+)]$ just right of its dominator
in the vertex ordering. More importantly,
blocks in a cluster do not have purple edges, as opposed to
a use of the consecutive method.  On the other hand, the nested
method   for super-blocks with the same dominator
does not allow  to save one page for purple edges at the composition,
as pointed out after Theorem~\ref{thm:10-pages}.

Similar to Lemmas 1 and 2 in \cite{y-epg4p-89}, we obtain a correspondence
between a partition induced by a super-block or a outer (edge) chord
and the vertex ordering, using Corollary~\ref{cor:chord-no-left} for the
right boundary.

\begin{lemma} \label{lem:ordering}  
Let $v_i$ and $v_k$ be outer vertices
of a   normalized planar 2-level multigraph $H$,
such that $v_i$ is the dominator and $v_k$ the last outer vertex
of a  super-block $B^+$ or $\edge{v_i}{v_k}$ is an outer   chord or an outer edge.
Let $H$
partition into parts $H_1$ and $H_2$  as described in Lemma~\ref{lem:partition},
and let $L$ be the set of vertices from super-blocks in $H_1$ that are dominated by $v_i$.
Then the interval
 $[v_i, v_k]$   contains  exactly the vertices of $L$ and  of $H_2$ with  $L < V(H_2)$,
that is we have $v_i, L,V(H_2),v_k$ in the vertex ordering.
\end{lemma}


Consider super-block $B^+$ with $\alpha(B^+)=v_i$ and $\omega(B^+)=v_k$
or an outer chord $\edge{v_i}{v_k}$.
If edge $e$ spans vertex $v_i$  and not $v_k$, then
$e$ is the first or last edge of the least super-block
$A^+$ that is  dominated by $v_i$ or $e$ crosses $\edge{v_i}{v_k}$.
 If edge $e$ spans vertex $v_k$  and not $v_i$,
then $e$ is incident to a vertex in $L$, since $v_k$ is a
boundary for edges with one vertex in $H_2$, as shown
in Lemma~\ref{lem:no-span}.
Now $e$   is an inner or a forward binding edge, since outer  chords and
forward binding edges do not span exactly one of $v_i$ and $v_k$.

  Next we consider also crossed edges.

\begin{definition}  \label{def:interfere}
Let $H^{\times}$ be   a normalized   1-planar 2-level multigraph.
Let $H_1$ and $H_2$ be parts of a decomposition of
the planar skeleton $H=H^{\times}[E_b]$ induced by
(the root of)    super-block $B^+$
or an outer (edge) chord $\edge{v_i}{v_k}$, as described in Lemma~\ref{lem:partition},
with  $v_i=\alpha(B^+)$ and $v_k=\omega(B^+)$.
Any two (uncrossed or crossed) edges $e_1$ and $e_2$  of $H^{\times}$
are \emph{well-separated}  if
  both vertices of $e_2$ are in $H_2  \cup \{v_i, v_k\}$ and both
vertices of $e_1$ are in $H_1 \cup \{v_i, v_k\} -L$,
where $L$ is the set of vertices from super-blocks dominated by $v_i$ in $H_1$.

 Otherwise, $e_1$  and $e_2$  \emph{may interfere}.
\end{definition}

\begin{lemma} \label{lem:interfere}  
Any two well-separated  edges can be embedded in the same page.
\end{lemma}
\begin{proof}
Partition  $H$ as described in  Lemma~\ref{lem:partition}.
If $e_1$ and $e_2$ are well-separated, then
 the vertices of $e_2$ are between and including $v_i$ and $v_k$ by
Lemma~\ref{lem:ordering},
and the vertices of $e_1$ are outside the interval
$[v_i, v_k]$.
Hence, $e$ and $e'$ nest or are disjoint.
\end{proof}

Hence, the correctness of a book embedding
can be reduced to the special case where
  edges may interfere. Clearly, two edges may nest or are disjoint if they may interfere.
Obviously, edges incident to vertices of super-blocks $A^+$ and $B^+$
are well-separated in the planar case, if there is a super-block $C^+$ in between
on a path of a block-tree, and similarly for a path with super-blocks $A^+, C^+, D^+$
and $B^+$ with crossed edges in a normalized 1-planar 2-level graph, since
there is a partition induced by $C^+$.

\subsection{Edge Coloring} \label{sect:edgecoloring}

In this section, let $H^{\times}$ be a 2-level subgraph of a  normalized
1-planar multigraph with outer cycle $O=v_0,\ldots, v_t$ and
some block $B=b_0,\ldots, b_q$.
A binding edge $\edge{b_i}{u}$ is called a \emph{pre-cluster edge}
if there is a cluster  $C(b_i, v)$ for $i=1,\ldots, q$
and $u<v$, see Figure~\ref{fig:pre-cluster}. Note that $b_i$ is the
leader of the cluster and there is a multi-edge $\edge{b_i}{v}$.\\

We  use the following \emph{coloring scheme}  that is adjusted to the later
page assignment.

\begin{figure}[t]  
  \centering
  \subfigure[]{
    \includegraphics[scale=0.47]{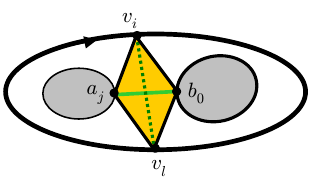}   
     \label{fig:last-1}
    }
 \subfigure[]{
    \includegraphics[scale=0.47]{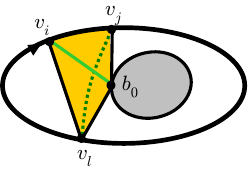}
\label{fig:last-2}
    }
\subfigure[]{
   \includegraphics[scale=0.47]{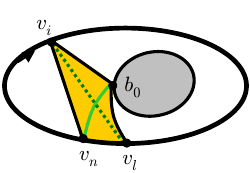}
\label{fig:last-3}
    }
\subfigure[]{
    \includegraphics[scale=0.47]{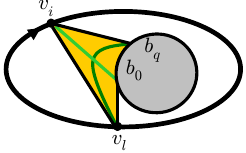}
\label{fig:last-5}
    }
\subfigure[]{
    \includegraphics[scale=0.47]{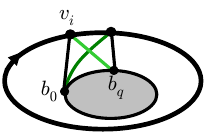}
\label{fig:last-4}
    }
\caption{Green crossed edges.
A quadrangle for the first face of a block-tree with first outer vertex
$v_i$ and last outer vertex $v_l$ is colored yellow.
A crossed outer chord is drawn dotted in
(a)-(c)  and a forward and a backward binding edge cross in
a quadrangle with the last edge of block $B$ in (d)-(e).
The first face is shaded.
  }
  \label{fig:last}
\end{figure}

(1) All edges of the planar skeleton are colored black,
except if edge $e$ is a  pre-cluster edge,
in which case $e$ is colored red.  

All other edges
cross in the interior of a quadrangle $Q$ with one to four outer vertices,
such that $g$ crosses $h$. It may be   the first face for a block-tree.

(2) If $Q$
has four outer vertices, then two outer chord cross. Now the outer chord
incident to the least outer vertex is colored green and red the other,
see Figure~\ref{fig:levelgraph}.
In all other cases, a crossed outer chord is colored green.

(3) If $Q$ has three outer vertices, then
   the  outer  chord $g$ crosses a binding edge $h$.
 Now both edges are colored green, see Figures~\ref{fig:last-2}
and \ref{fig:last-3}.

(4)  If $Q$ has two inner vertices, then (a) either an outer chord
crosses an inner chord that connects two   block-trees,
 see Figure~\ref{fig:last-1}. Now both edges are colored green.
Or (b) two binding $g$ and $h$ edges cross. Now both edges
are colored green if $Q$ contains the last edge of a block,
see Figures~\ref{fig:last-5} and \ref{fig:last-4}. 
Otherwise, (c)   $g$ is colored green and $h$ red if  $g \prec h$, that is $g$
 is incident to the least outer vertex of $Q$.

(5) If $Q$ has a single outer vertex, then a binding edge $g$
crosses an inner chord $h$,   where the
  vertices of $h$ are in one or two blocks of a block-tree.
Now  $g$ is colored green.
(a) Edge $h$ is   colored red   if both vertices are in the same block,
or  (b) $Q$ contains the first edge of an uncovered block,
or (c) one vertex of $h$ is in a covered block $B$ and the other vertex is
the leader of $B$, such that $h=\edge{b_0}{b_2}$ or $h=\edge{b_0}{b_{q-1}}$.
(d)  Edge $h$ is colored purple, as shown in Figure~\ref{fig:badedges},
 if its vertices are in two uncovered blocks and $Q$
contains the last edge of a block, where the (multi-) edge of a
 block with two vertices is its last edge, see Figure~\ref{fig:path}.\\

\begin{figure}[t]   
  \centering
  \subfigure[]{
    \includegraphics[scale=0.6]{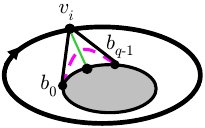}   
    }
 \subfigure[]{
    \includegraphics[scale=0.6]{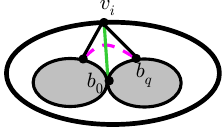}
    }
\subfigure[]{
    \includegraphics[scale=0.6]{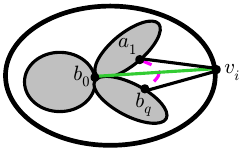}
    }
\caption{ A purple edge assigned to block $B$, called  a
(a)   handle  (b)
 a connector and  (c)  a  bridge,  respectively.
  }
  \label{fig:badedges}
\end{figure}

\begin{figure}[t]
\centering
\includegraphics[width=0.4\textwidth]{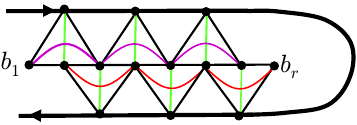} 
  \caption{A path of blocks with inner chords on either side.
  }
  \label{fig:path}
\end{figure}

This covers all cases, since $H^{\times}$ is normalized.
Recall that there are no inner chords with vertices in two blocks, such as
a connector or a bridge,
if the blocks are  covered and are in a cluster.
Also all outer chords are  black or green, except if two outer
chords cross, and similarly for the binding edges.
where  two green edges may cross.
On the other hand, all inner chords are red or purple.

\subsection{Embedding  Black and Green  Edges}\label{subsect:edgeassignment}

We extend Yannakakis~\cite{y-epg4p-89} algorithm
 from blocks to super-blocks. There is the same vertex ordering and page
assignment for all uncrossed edges if there are no covered blocks.
As observed before, the vertices on the outer cycle $O$ are consecutive
 if $O$ is traversed by the consecutive method at
the previous level, except for its leader.
They are distributed over intervals if the
nested method was used, for example if $O$ is in a cluster.
As observed in \cite{y-epg4p-89}, this
does not matter for the correctness of
the book embedding, since there are no edges between vertices
from the interior of $O$ and $O\rq{}$ if $O$ and $O\rq{}$ are
blocks at the previous level, and the intervals for the vertices of
 $O$ and $O\rq{}$ are either internally disjoint or they nest.

\begin{lemma} \label{lem:planarpages}  
The black and green edges of a normalized   planar 2-level multigraph
can be  embedded in three pages $\eta(B^+), \pi(B^+)$ and
$\overline{\pi}(B^+)$  (or $P_1, P_2, P_3$)
if the vertex ordering $L(H)$ is used,
such that the edges of any super-block $B^+$  are embedded in page
$\pi(B^+)$,  the forward binding edges incident to vertices of $B^+$
in page $\overline{\pi}(B^+)$, and all outer edges, outer
chords and backward binding edges
in page  $\eta(B^+)$, where $\eta(B^+)=\pi(O)$ for the outer cycle $O$.
In addition,  let $\eta(B^+) = \eta(A^+)$ for any two super-blocks and let
$\pi(B^+)=\overline{\pi}(A^+)$ if both are in the same expanded block-tree and
the distance between the root and $A^+$ is even and is odd for $B^+$,
or vice versa, and let $\pi(A^+)=\pi(B^+)$, otherwise.
\end{lemma}

\begin{proof}
Yannakakis  \cite{y-epg4p-89} has shown that there is no conflict in
any page by his page assignment of edges  if there is only one
block-tree and   there are no covered blocks. We use the same
vertex ordering and page assignment when restricted to vertices in uncovered
blocks and use the arguments for the nested method for inner edges and page
$\pi(B^+)$.
The extension to planar 2-level multigraphs with outer chords
and covered blocks is proved in the  same way,
using Lemmas~\ref{lem:partition} and \ref{lem:ordering} and induction on the number
of super-blocks and  outer chords.
Observe that one edge is embedded in page $\eta(B^+)$ and the other in
$\pi(B^+)$ or $\overline{\pi}(B^+)$ if two green edges cross,
as defined before in cases (3) and (4)(a) and (4)(b). Note that
any two such edges twist.

Concerning block-expansions, observe that the vertices of a cluster
$C(b_i,v)$ are placed in an interval just right of $b_i$.
Then all blocks of $C(b_i,v)$ are covered by $v$. All
inner edges are embedded in page $\pi(B)$, as in the nested method.
All binding edges
incident to a vertex of $C(b_i,v)$ are incident to $v$.
They form a fan at $v$ and do not twist mutually.  For any vertex
$c \in V(C(b_i,v))$, edge $\edge{c}{v}$ is backward
binding if $v$ also dominates the block of $b_i$, such as
vertex $v_k$ in Figure~\ref{fig:pre-cluster}. Then $\edge{b_i}{v}$
is a backward binding (multi-) edge. Now, all edges $\edge{b_j}{v}$
with $i\leq j \leq q$ are backward binding.
Edges $\edge{b_j}{v}$ and $\edge{c}{v}$ do not twist, since
$c$ is placed just right of $b_i$. Hence,
  all edges $\edge{c}{v}$ can also be embedded in page $\eta(B)$.
Similarly, if the block of $b_i$ is not dominated by $v$, such as
vertex $v_n$   in Figure~\ref{fig:pre-cluster}, then
edge $\edge{b_i}{v}$ is forward binding and is embedded in page
$\overline{\pi}(B)$. Now all edges  $\edge{c}{v}$ are forward binding
and are embedded in $\overline{\pi}(B)$. However, any pre-cluster edge
$\edge{b_i}{v\rq{}}$ twists any binding edge $\edge{c}{v}$. As
pre-cluster edges are colored red, they are embedded in page $\rho(B)$,
as shown in Lemma~\ref{lem:x1}.
\end{proof}

\subsection{Embedding Red and Purple Edges} \label{subsect:assignerededges} 

Yannakakis~\cite{y-epg4p-89} has observed that all forward binding edges
of a planar 2-level graph can be embedded in a single page.
Similarly,  the red edges of a normalized 2-level 1-planar
graph $H^{\times}$  can be embedded in a single page $\rho(B)$,
including all pre-cluster edges.
Moreover,  the last edges of all blocks are embedded in two pages that alternate
between a block and its parent, as stated in Lemma~\ref{lem:planarpages}.
Similarly, the purple edges can be embedded in pages
$\chi(B)$  and $\overline{\chi}(B)$   (or $P_5, P_6$), where $\chi(B)$ and
$\overline{\chi}(B)$ are opposite to each other, similar to $\pi(B)$
and $\overline{\pi}(B)$.
Now, covered  blocks are simpler than
uncovered ones, since  they have no incident purple edges.

First, observe that there is no red or purple edge with
vertices in two block-trees. Any two block-trees are separated by an uncrossed
outer chord or by a quadrangle that contains
an outer chord that is crossed by a green inner chord.
 By 1-planarity, a connection by  a red or purple edge is excluded.
Second, if block $B$ is covered, then it is in a cluster $C(a,v)$. Then
all binding edges incident to vertices of $B$ are incident to $v$
such that all  red binding edges    form a fan at $v$.  Similarly, all pre-cluster
edges form a fan at the leader $a$ of the cluster:

\begin{lemma} \label{lem:x1}  
All   red    edges
of a normalized 1-planar 2-level  graph $H^{\times}$ can be embedded in a single page.
\end{lemma}
\begin{proof}
An edge is red if it is (1)  a pre-cluster edge or (2)
  an outer chord  that is crossed by another outer chord,  or (4) a binding edge
that   is 
 crossed by a binding edge in a quadrangle without the last edge
of an uncovered   block 
 or  (5) an inner chord, except if it is  in a quadrangle $Q$ with
 the last edge of an uncovered   block, in which case it is purple
We proceed by induction and consider super-blocks
in the order of their  dominator.
 By Lemma~\ref{lem:interfere}, it suffices to consider
red  edges that may interfere, which is not used in the proof.

\begin{figure}[t]   
  \centering
  \subfigure[]{
    \includegraphics[scale=0.9]{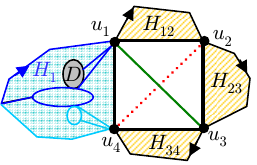}  
    \label{fig:cross-outer}
}
\hspace{4mm}
\subfigure[]{
    \includegraphics[scale=0.9]{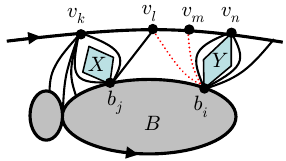}  
    \label{fig:pre-cluster}
}
\caption{Illustration for the proof of Lemma~\ref{lem:x1}.
(a) A pair of crossed outer chords and the induced partition of $H$.
(b) Block $B$ is dominated by vertex $v_k$ and is expanded by clusters
clusters $X=C(b_j, v_k)$
and $Y=C(b_i, v_n)$ and pre-cluster edges $\edge{b_i}{v_l},
\ldots, \edge{b_i}{v_m}$.
  }
  \label{fig:pre}
\end{figure}

First, we show that red outer chords can be embedded in single page, denoted
$\rho(B)$ (or $P_4$), and that a pair of crossed outer chords
induces a  partition of $H$ and of $L(H)$.
A red outer chord $e = \edge{u_2}{u_4}$
is crossed by a green outer chord $\edge{u_1}{u_3}$
such that $\edge{u_i}{i_{i+1}}$ is an uncrossed outer chord or an outer edge
by the normalization for $i=1,2,3,4$ and $u_5=u_1$, see Figure~\ref{fig:cross-outer}.
It may be a multi-edge, such that $\edge{u_i}{i_{i+1}}$ is a separation pair.
Each of them induces a partition of $H$ by  Lemma~\ref{lem:partition}.
Denote the second part by $H_{i,i+1}$ for
$i=1,2,3,4$, where index computations are modulo four,
and let $H_1$ denote the first part for $\edge{u_1}{u_4}$.
The second part  contains all inner components
if $\seppair{u_i}{u_{i+1}}$ is a separation pair.
Vertex $u_1$ may dominate  super-blocks in $H_1$,
whose vertices are   placed just right of $u_1$.
Vertex $u_{2}$ does not dominate any  super-block  in part  $H_{12}$
if $H$ is partitioned by $\edge{u_1}{u_{2}}$ by Corollary~\ref{cor:chord-no-left}.
 Hence, there are no
vertices in $H_{12}$   to the right of $u_2$, and similarly
 for $u_3$ and $u_4$. In consequence, exactly the vertices of $H_{23}$
are between $u_2$ and $u_3$ by Lemma~\ref{lem:ordering},
and exactly the vertices of $H_{34}$ are between $u_3$ and $u_4$.
Thus the set of vertices $V$ satisfies
$V_l < u_1 < L < V(H_{12}) <u_2 < V(H_{23}) < u_3
< V(H_{34}) < u_4 < V_r$,
  where $V(H_1)=V_l \cup L \cup V_r$, and $L$ is the set of vertices
from super-blocks in $H_1$ that are dominated by $u_1$.

There is no edge between vertices of the first and the second part from
a partition induced by any of $\edge{u_i}{u_{i+1}}$, since it is uncrossed.
Hence, there is no edge between any two of $H_1,  H_{12},
 H_{23}$ and $H_{34}$.
There is no edge between $u_1$ and a vertex of $H_{23}$ or $H_{34}$,
since $\edge{u_2}{u_3}$  and $\edge{u_3}{u_4}$  are uncrossed outer  chords
Similarly, there is no edge between $u_3$ and $H_1$.
However, edges  $\edge{u_1}{u_3}$ and $\edge{u_2}{u_4}$ cross in the
1-planar drawing and twist in $L(H)$.
Hence, we obtain a partition of $H$ and its linear ordering $L(H)$, such that
any edge with one vertex in $H_{12}$ has both vertices in $H_{12}\cup \{u_1,u_2\}$,
and similarly for $H_{23}$  and $H_{34}$. In particular, these edges
do not twist the red outer chord $\edge{u_2}{u_4}$. Hence, only edges with a vertex
in $H_1$ may cause, trouble, in particular, those with a vertex in $L$.
However, if edge $e$ has one vertex on $H_1$,
then its other vertex is in $H_1 \cup \{u_1,u_4\}$, that that
$e$ and $\edge{u_2}{u_4}$ are disjoint or nest. Hence, the
red outer chord $\edge{u_2}{u_4}$ can be embedded  in page $\rho(B)$,
and also in $\pi(B)$ and $\overline{\pi}(B)$, where an embedding
in $\pi(B)$ has side effects at the next level,
 when   $B$ is the outer cycle
and may have an outer chord incident to $b_{i+1}$ in its interior.
Observe that  its crossing green outer chord $\edge{u_1}{u_3}$
can only be embedded in the page for the outer edges  $\eta(B)$.
By  induction,  all red  outer chords can be embedded in a single page.

Forthcoming, it suffices to consider red  edges with inner vertices
in a single block-tree, since there is a partition, as observed before.
In addition,  a  red inner  chord $e$ with both vertices in the same
uncovered block cannot twist any other red  edge, similar to the
case of green inner chords with vertices in a covered block.
Namely,  $e=\edge{b_i}{b_{i+2}}$ for $B=b_0,\ldots,b_q$ and $1\leq i \leq q-2$,
such that  $b_{i+1}$ is not a cutvertex,
as for block $A$ in Figure~\ref{fig:levelgraph}.
 Then the edges incident to
$b_{i+1}$  is colored green, such that edge $e$ can be embedded  in
$\rho(B)$.

Observe that the red edges incident to vertices of a single block-tree
are ordered by the faces that contain them, similar to
forward binding edges in the planar case \cite{y-epg4p-89}.
Using Lemmas~\ref{lem:partition} and \ref{lem:ordering}, or the
arguments by Yannakakis~\cite{y-epg4p-89}, we obtain that any two red
binding edges incident to vertices of uncovered blocks
do no cross. It is due to the fact that blocks and the outer cycle
are traversed in opposite directions. Next, all pre-cluster edges
can be added to page $\rho(B)$.
To see this, consider all clusters $C(b_i, v)$ at $b_i$
with different vertices $v$. The  vertices of each cluster
from   the interval  in the interior of $[b_i, b_{i+1}]$.
The pre-cluster edges incident to $b_i$ and outer
vertices from the clusters  form a fan at $b_i$, such that
they do not twist mutually. As
the binding edges $\edge{c}{v}$
are colored green and are already embedded in page $\overline{\pi}(B)$,
any pre-cluster edge $\edge{b_i}{v}$ can be added to page $\rho(B)$,
such that it does not twist another edge in this page.

\begin{figure}[t]   
  \centering
  \subfigure[]{
    \includegraphics[scale=0.8]{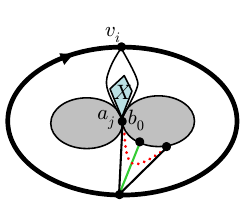}   
    }
\hspace{4mm}
 \subfigure[]{
    \includegraphics[scale=0.8]{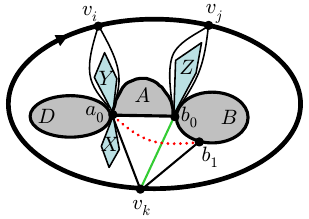}
    }
\caption{Illustration for the proof of Lemma~\ref{lem:x1}.
 A red inner chord between two blocks, drawn red dotted,
and clusters $X, Y$ at $a_0$ and $Z$ at $b_0$ with $b_0=a_1$.
Block $D$ is expanded at $a_0$ by $X$ and $Y$ and block $A$
at $a_1$ by $Z$.
  }
  \label{fig:redchords}
\end{figure}

It remains to consider red inner chords $e$ with vertices in two uncovered blocks
$A=a_0,\ldots, a_p$ and $ B=b_0,\ldots, b_q$, see Figure~\ref{fig:redchords}.
These are the root of super-blocks $A^+$ and $B^+$,
respectively. Let $A^+ \prec  B^+$.
Then $A^+$ is the parent of $B^+$ and
$B^+$ is the last child at a cutvertex $a_j$ of $A$
such that $e= \edge{a_j}{b_2}$ or $e=\edge{a_{j-1}}{b_1}$,
where  $b_0=a_j$ is the leader
of $B$ and the cutvertex between $A$ and $B$ for $j=1,\ldots, p$.
Otherwise, edge $e$ is colored purple,
since it is in a quadrangle containing the last edge of $B$.
Assume $j=1$, as the case $j >1$ is simpler.
Then vertex $a_0$ is in a block $D$ preceding  $A$, and it is not the
root of $D$. Blocks $D, A$ and $B$ form a path and there
may be clusters at $a_0$ and $a_1$.
Assume that  the dominators of $D, A$ and $B$ are different.
Otherwise, there are backward binding edges that do not disturb the
embedding of red edges.
For $D=d_0,\ldots, d_r$  and $a_0=d_h$ for $1 \leq h \leq r$ we obtain

$d_{h-1}, a_0, C_1, d_{h+1}, \ldots, d_r, \ldots, \alpha(A), a_1,C_2, a_2, \ldots, a_p,
\ldots, \alpha(B), b_1, b_2$

\noindent in the vertex ordering, where
$C_1$ and $C_2$ are the sets of vertices form all clusters $C(a_0, v)$ at $a_0$
and $C(a_1, v\rq{})$ at $a_1$, respectively.
Then $\alpha(A) \leq v\rq{} \leq  \alpha(B)$, since the quadrangle
$Q=a_0, a_1, b_1, v_s$ prevents clusters on the other side with
an outer vertex greater than $\alpha(B)$.

We claim that the red edge $\edge{a_0}{b_1}$ spans all red
edges $\edge{x}{y}$ with $a_0 < x < b_1$. The case for $\edge{b_0}{b_2}$
with $b_0=a_j$ and $j \geq 1$ is similar.
Note  that $\alpha(A)$ dominates only the super-block $A^+$
with root $A$, and similarly for $\alpha(B)$ and $B^+$.
Otherwise, if $\alpha(A)$ dominates block $X$, then $X$ is
covered by $\alpha(A)$, since $X$ is in the interior of a curve from
$\alpha(A)$ through $\lambda(X), \lambda(A)$ and $a_p$.
Then $X$ is covered by $\alpha(A)$ and is merged into another block.
In consequence, $\alpha(A)$ does not dominate blocks in part $H_1$ from the
partition induced by $A$, such that the set $L$ is empty in
Lemma~\ref{lem:ordering}, and similarly for $B$.
Hence, (the vertices of) edges incident to (vertices of) $A$ and $D$ are in part
$H_1$ from a partition induced by $B$, such that  they are outside the
interval $[\alpha(B), \omega(B)]$ by Lemma~\ref{lem:ordering}.
Clearly $b_1< \omega(B)$. In consequence, all forward binding edges
with an inner vertex   between and including $a_0$ and $\alpha(B)$ have both vertices
in the interval $[a_0, \alpha(B)]$, and similarly for inner edges. All these
edges are spanned by the red inner chord $\edge{a_0}{b_1}$, such that they
can be embedded in the same page. By induction, all red inner chords with
(or without) binding edges in two blocks can be embedded in page $\rho(B)$.

By induction on the number of block-expansion in blocks of $\mathcal{T}$, all
red  edges  edges can be embedded in page  $\rho(B)$.
\end{proof}

Note that   binding edges incident to a vertex from an inner
component at an   outer separation pair
$\seppair{u}{v}$ with $u<v$ can be recolored green, since the
  binding edges incident to $u$ can be taken as backward binding and
are embedded in page $\eta(B)$ and those incident to $v$
as forward binding with an embedding in $\overline{\pi}(B)$, or vice versa. \\

The last edges of blocks cause trouble in the planar case, as they
span all other vertices of their block and the dominator, in general.
In consequence, two pages $\pi(B)$ and $\overline{\pi}(B)$ must be used for the  embedding of the last edges from blocks $A$ and $B$ if $A$ is the parent of $B$.
A purple edge is similar, as it spans the last edge of a block
in the 1-planar drawing and the linear ordering,
see Figure~\ref{fig:all-dashed-pages}.
A purple edge is in a quadrangle $Q$ of $H{\times}[E_b]$
containing the last edge of a non-elementary,
uncovered  (super-) block $B^+$ and its dominator, as
shown in Figure~\ref{fig:badedges}. It is an inner chord
with vertices in two super-blocks $A^+$ and $B^+$ from a block-tree
or only
in one super-block if it is a root.
There are no purple edges for covered blocks.
Hence,   block-expansions can be disregarded for purple edges.
By the one to one correspondence between uncovered blocks
and super-blocks, 
there is at most one purple edge per
super-block. Similar to the first and last edge of a block,
let   $h(B^+)$ denote the \emph{purple edge} of   super-block
  $B^+$ if the last vertex of  $h(B^+)$ is in $B^+$.
There are four more cases for a quadrangle $Q$ containing the last
edge of a super-block and the leader of a root,
 as shown in Figures~\ref{fig:last}.
In these cases, an outer chord crosses   an inner chord or binding edge or
two binding edges cross, one is  backward and the other is forward binding.
These edges are colored green and are embedded in pages
$\eta(B^+)$ and $\overline{\pi}(B^+)$, respectively,
as shown in Lemma~\ref{lem:planarpages}.

In particular, if there is a path $b_1,\ldots, b_r$
of uncovered blocks $B_i=b_i, b_{i+1}$ with a (multi-) edge in between
and quadrangle with only one outer vertex in $H^{\times}[E_b]$,
then the crossed inner chords are purple on the left side
and red on the other side, see Figure~\ref{fig:path}.
 In this special case, there is no need to alternate between a block
and its parent, such that the path can be treated like
a single block. Observe that there is no path with two or
more covered blocks by the normalization.

Forthcoming, it suffices to consider non-elementary and uncovered blocks
$B = b_0,\ldots, b_q$ and $A=a_0,\ldots, a_p$ that are the root of super-blocks
$B^+$ and $A^+$, respectively.
The purple edge  assigned to $B$ is called
  a \emph{handle} if $h(B)= \edge{b_0}{ b_{q-1}}$,
a \emph{connector}
if $h(B)=\edge{b_q}{a_{j+1}}$, where $b_0=a_j$ is in  block $A$
and $a_{j+1}=a_0$ if $a_j$ is the last vertex of $A$,
a  \emph{bridge}
if $h(B)=\edge{a_1}{b_q}$ if $A$ is the sibling that precedes $B$ at
 the common cutvertex $a_0=b_0$, and
a \emph{hook} if $B$ is the root of a block-tree $\mathcal{T}$
and $h(B)= \edge{b_q}{\omega(\mathcal{T})}$, see Figure~\ref{fig:badedges}.

\begin{lemma} \label{lem:x3} 
All purple  edges of a normalized 2-level 1-planar graph
can be embedded in   pages  $\chi(B)$ and $\overline{\chi}(B)$, respectively.
\end{lemma}
\begin{proof}
The proof  is similar to the one
of Lemma~\ref{lem:planarpages}, where pages $\pi(B)$ and $\overline{\pi}(B)$
are used alternately, in particular, for last edges of
uncovered blocks or super-blocks that share a cutvertex.
Block $B$ is non-elementary and uncovered if it has a purple edge $h(B)$,
and it is the root of a super-block $B^+$, such that $h(B)=h(B^+)$. Hence, its
vertices are in an interval just right of the dominator.
 It does not matter
that the outer cycle may be traversed by the nested method at the
previous level, such that its vertices are contained in an interval that
also contains sub-intervals for other blocks.

We claim that purple edges $h(A)$ and $h(B)$ twist if and only if
$A$ is the parent of $B$ and $d(B)$ is a handle
or $B$ is the first child of $A$ at $b_0$ and $d(B)$ is a connector
or $A$ is the immediate predecessor of $B$ at a common cutvertex, such
that $h(B)$ is a bridge, see Figure~\ref{fig:badedges}. To see this,
observe that
blocks $A$ and $B$ are the root of super-blocks $A^+$ and $B^+$.
They are in the same block-tree if they twist, since
two block-trees are not connected by a red or purple edge,
and they are  they are separated by an outer chord.
Using Lemmas~\ref{lem:partition} and \ref{lem:ordering}, $h(A)$
and $h(B)$ are well-separated or the  vertices of $h(B)$ are  in and
those of $h(A)$ are outside an interval, otherwise. This holds, in particular
if $h(A)$ and $h(B)$ are connectors and there is a path of (super-) blocks
$A, C, B$ or if $h(A)$ and $h(B)$ are bridges and (super-) blocks
$A<C<B$ share a common cutvertex or if $h(A)$ and $h(B)$ are
handles and the leader of $A$ and $B$ is in the same block.
Hence,
there are no other cases for twisting purple edges.

The vertices of uncovered blocks are placed one after another. Vertices
from clusters and outer vertices do not matter, since purple
edges are inner chords.
Let $A\rq{}$ and $B\rq{}$ be the (super-) blocks containing the least
vertex of $h(A)$ and $h(B)$, respectively. The quadrangle containing $h(A)$
precedes the quadrangle containing $h(B)$ in the planar skeleton $H^{\times}[E_b]$.
In consequence, orderings
$A\rq{} \prec B\rq{} \prec A \prec B$ and $B\rq{} \prec A\rq{} \prec A \prec B$
are excluded. Edges $h(A)$ and $h(B)$ are disjoint if
$A\rq{} \prec A  \prec B\rq{} \prec B$. Then only $A\rq{} =B\rq{}$
or $A=B\rq{}$ remain,
since $A \neq B$ as  each (super-) block has at most one purple edge.
If $A\rq{}=B\rq{}$, then
$\lambda(B) \leq \lambda(A)$,  $h(A)$ and $h(B)$ nest,
since $x \leq y$ where $x$ and $y$ are are the least vertices
of $h(B)$ and $h(A)$, respectively.
Otherwise, $h(A)$ and $h(B)$ twist if $A=B\rq{}$, such that the purple
edges are embedded in two pages.
If there are three or more purple edges of blocks $A,B$ and $C$,
then at least two of them are disjoint or nest by the previous case analysis,
such that two pages suffice for the embedding of all purple edges.
\end{proof}

\begin{figure}[t]
  \centering
  \subfigure[]{
    \includegraphics[scale=0.7]{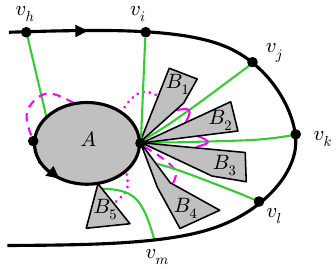} 
    \label{fig:purple-graph}
    }
\subfigure[]{
    \includegraphics[scale=0.8]{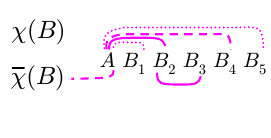} 
    \label{fig:purple-pages}
    }
\caption{Illustration for the proof of Lemma~\ref{lem:x3} with (a) handles
(dashed), connectors (dotted), and bridges (solid)
and (b) their embedding in pages $\chi(B)$  and $\overline{\chi}(B)$,
respectively.
  }
  \label{fig:lemx3}
\end{figure}

From Lemmas~\ref{lem:planarpages}, \ref{lem:x1} and  \ref{lem:x3} we
obtain:

\begin{theorem} \label{thm:2-level}
Any 2-level subgraph of a normalized 1-planar graph admits a 6-page book
embedding, such that all
  black, green and red dotted edges are embedded in four pages
and the purple
ones in two pages. In addition,
  the edges of the outer cycle and the edges of any (super-) block are
embedded in a single page.
\end{theorem}

\section{Embedding  1-Planar Graphs} \label{sect:main}
By the peeling technique, the 6-page book embedding  of  2-level 1-planar
multigraphs can be composed to a 12-page book embedding of  1-planar graphs.
  One page can be saved  for the
uncrossed edges as observed by Yannakakis \cite{y-epg4p-89} for his
5-page algorithm for planar graphs, since the page for the edges of a (super-) block
is  reused for the outer chords and backward binding edges at the next level. Similarly,
a page for the purple edges can be reused  at the next level.

\begin{lemma} \label{lem:dashedred-and-last-chords} 
If an  outer chord is incident to the last outer vertex, 
then, as an alternative to $\eta(B)$,  it can  be embedded in one of $\pi(B)$ or
$\overline{\pi}(B)$ (or in one of $\chi(B)$ or $\overline{\chi}(B)$),
 such that there is no conflict at all previous  or later levels.
\end{lemma}

\begin{figure}[t]
  \centering
  \subfigure[]{
    \includegraphics[scale=0.6]{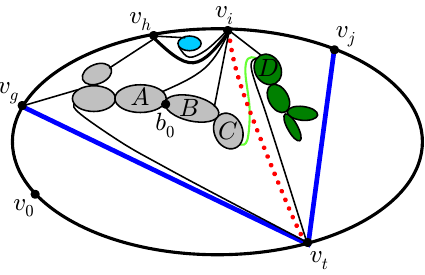}
    \label{fig:lastcord-graph}
    }
 \subfigure[]{
    \includegraphics[scale=0.7]{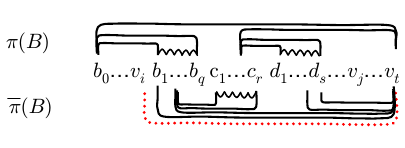}  
    \label{fig:lastchord-pages}
    }

\caption{Illustration for the proof of Lemma~\ref{lem:dashedred-and-last-chords}
 (a) with outer chords
incident to the last outer vertex $v_t$ and (b) edges embedded in pages $\pi(B)$
and $\overline{\pi}(B)$, respectively.
  }
  \label{fig:lastchord}
\end{figure}

 \begin{proof}
An  outer chord $\edge{v_i}{v_t}$ in the interior of $O=v_0,\ldots, v_t$
is embedded in page $\eta(B)=\pi(O)$ by Lemma~\ref{lem:planarpages}.
It
partitions the planar subgraph $H[E_b]$
into parts $H_1$ and $H_2$ by Lemma~\ref{lem:partition} if it is uncrossed.
There is a similar partition induced by the root of the block-tree to the right
of  $\edge{v_i}{v_t}$ if it is crossed by a green inner chord $\edge{x}{y}$
between two block-trees, see Figure~\ref{fig:last-1}.
We claim that   $\edge{v_i}{v_t}$ can also be embedded in page $\overline{\pi}(B)$, where
$B$ is the least super-block dominated by $v_i$ in part $H_1$, $B=x,y$ if
$v_i$ does not dominate a super-block in $H_1$ but $\edge{x}{y}$ crosses $\edge{v_i}{v_t}$,
and $B$ is arbitrary, otherwise.
By Lemmas~\ref{lem:partition}  and \ref{lem:ordering}, we have
 $b_0 <v_i < L < V(H_2) < v_t$, where $b_0$ is the leader of $B$ in some block $A$,
block $B$ is the first super-block of $L$, and $V(H_2)$ is the set of vertices in part $H_2$
induced by the outer chord $\edge{v_i}{v_t}$.
Then the first and last edge of $B$ span $v_i$. These edges are embedded in page
$\pi(B)$. This also holds if $B=x,y$.
By the use of block-expansions, if $\edge{d}{u}$ is a forward binding edge with $d$ in part
$H_1$, then $u \leq v_i$ or $u=v_t$, since $\edge{v_i}{v_t}$ is an outer
chord and $v_t$ is last.
Clearly, there is no edge in $H$ that spans $v_t$, since $v_t$ is last.
If $v_t$  dominates a block, the it covers it, such that the
dominated  block is merged into another block by a block-expansion.

 Hence, vertices $v_i$ and $v_t$ are ``free'' in page $\overline{\pi}(B)$,
such that
edge  $\edge{v_i}{v_t}$ can be added to $\overline{\pi}(B)$ without   twisting
another edge embedded in the page. It does not create a conflict at the next level,
since page $\pi(B)$ is used for outer chords and backward binding edges at the next level
in the interior of $B$, and there is no conflict from previous levels,
since both $v_i$ and $v_t$ are  free.
If there are outer chords $\edge{v_i}{v_t}$ and $\edge{v_j}{v_t}$, then let
$\pi(B)=\pi(B\rq{})$
or $\pi(B)=\overline{\pi}(B\rq{})$, such that both pages $\pi(B)$ and $\overline{\pi}(B)$
are used for the re-embedding. Alternatively, one can
synchronize the pages. For outer chords $\edge{v_i}{v_t}$ and $\edge{v_j}{v_t}$,
let $B_i$ and $B_j$ be the least blocks dominated by $v_i$ and $v_j$, respectively.
Let $B_i < B_j$ and assume that $\pi(B_i) = \overline{\pi}(B_j)$. Then swap pages
$\pi(B)$ and $\overline{\pi}(B)$ for all blocks in the block-tree containing $B_j$.
This is doable, even if an  outer chord  incident to $v_t$ is  crossed by an edge $\edge{a}{b}$,
since $b$ is the first vertex of the  root. Then all outer chords incident to
the last outer vertex are re-embedded in page  $\overline{\pi}(B)$.

  The case for pages $\overline{\chi}(B)$ and $\chi(B)$ is similar.
\end{proof}

 We wish to embed a connector $\edge{a_{j+1}}{b_q}$ or a bridge
$\edge{a_1}{a_q}$ of block $B$ in another page, in a similar way as before for the
last outer edge. Both are incident to the last vertex of $B$ and close to its
last edge $\edge{b_0}{b_q}$, where $b_0=a_j$ for a parent block $A$,
see Figure~\ref{fig:purple}.
If a connector is embedded in page $\pi(B)$, then the first edge
$\edge{b_0}{b_1}$ must be embedded in another page.
Moreover,  there is a problem at the
next level, since all inner chords $\edge{b_0}{b_j}$
that are embedded
in page $\pi(B)$, too, see Lemma~\ref{lem:planarpages}.
As there are outer vertices between $b_0$ and $b_1$, there will be
a conflict with inner chords  $\edge{b_0}{b_j}$ in
 any of $\eta(O), \pi(O), \overline{\pi}(O)$ and $\rho(O)$.
Thus there are no means to save pages
$\chi(B)$ and $\overline{\chi}(B)$.

\begin{theorem}  \label{thm:10-pages}
 Every   1-planar graph  admits a  10-page book embedding
that can be computed in linear time from a  1-planar drawing.
\end{theorem}

\begin{proof}
Augment the 1-planar drawing to a normalized 1-planar multigraph $G$
according to Lemma~\ref{lem:normalform}. The added edges are
  removed in the end.
 The leveling is computed iteratively from the
planar skeleton $G[E_b]$. The vertex ordering of $G$ is obtained from the vertex
ordering of  its planar 2-level subgraphs,
as described in Section~\ref{sect:vertex-ordering}. By Theorem~\ref{thm:2-level},
any normalized 1-planar 2-level graph has a 6-page book embedding.
We will save two pages at the composition of the embedding
of 1-planar 2-level subgraphs at any two consecutive or all odd and even levels
without changing the vertex ordering.
We have $\eta(B) = \pi(O)$ and use pages $\eta(B), \pi(B), \overline{\pi}(B),
\pi(X)$ and $\overline{\pi}(X)$ for nested blocks $O, B$ and $X$
at levels $\ell, \ell+1$ and $\ell+2$, respectively,
which is Yannakakis 5-page algorithm for planar graphs.
  In addition, we use pages $\rho(B)$ and
 $\rho(X)$ for the red  edges
and pages $\chi(B), \overline{\chi}(B)$ and $\chi(X)$ for the purple ones.

\begin{figure}[t]
  \centering
  \subfigure[]{
    \includegraphics[scale=0.7]{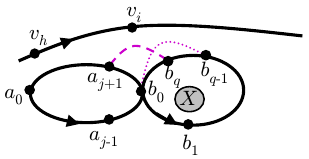}                 
    \label{fig:all-dashed-graph}
    }
    \hspace{1mm}
    \subfigure[]{
        \includegraphics[scale=0.7]{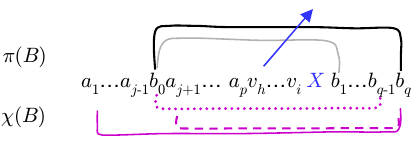}
        \label{fig:all-dashed-pages}
} \caption{Illustration for the proof of Theorem~\ref{thm:10-pages}.
(a) Purple edges in a graph,  a  connector $\edge{a_{j+1}}{b_q}$, drawn purple and
dashed and a handle, drawn purple and dotted. (b) A connector and a bridge,
 drawn purple and solid in a page embedding. Block $X$ is from the next level
and the blue arrow indicates edges from the previous level incident to
vertices $v_h,\ldots, v_i$.
  }
  \label{fig:purple}
\end{figure}

It remains to show that   page $ \chi(X)$ can be  saved.
Recall that there are only a few purple edges as each such edge is assigned to
an uncovered block  or  a super-block  and is a
   handle,  a  connector,  and a bridge, respectively,
see Figure~\ref{fig:badedges}.

If $h(B)$  is   a connector or a bridge, then it spans the vertices
from any block $X$ in the interior of $B$, since $X$ is placed
immediately to the left of $b_1$ if it is dominated by $b_0$ and
between $b_1$ and $b_q$, otherwise.
Hence,   page $\chi(B)$ can be reused for the edges
assigned to page $ \chi(X)$.

If $h(B)$ is a handle  $\edge{b_0}{b_{q-1}}$ of block $B=b_0,\ldots, b_q$,
then  we move it from $\chi(B)$ to $\pi(B)$ and move the chords incident to
the last vertex $b_q$ of $B$ in the interior of $B$
  from $\pi(B)$ to   $\overline{\pi}(B)$ (or  $\overline{\chi}(B)$),
as shown in  Lemma~\ref{lem:dashedred-and-last-chords}.
Handle  $\edge{b_0}{b_{q-1}}$ does not twist any edge
embedded on page $\pi(B)$. It nests with the last edge of $B$, and
possible conflicts at the next levels are avoided by moving chords.
Similar to the last edge of $B$, the handle does not twist
any edge embedded into page $\pi(B)$ at the previous level.
By induction on the blocks at level $\ell+1$,  all handles
are moved to pages $\pi(B)$ and $\overline{\pi}(X)$, respectively.
Hence, page
$\chi(B)$ can been saved, such that
all edges of $G$ are embedded in ten pages.

Concerning the running time, the computations on $G$ are performed
on a planarization of $G$, which is obtained from a
1-planar drawing  of $G$ in
linear time. The augmentation of $G$ to a planar-maximal graph,
the
 decomposition into 3-connected components at separation pairs, and
 the edge coloring can be computed in linear time.
The computation of the vertex ordering  takes linear time, both for
2-level graphs, and, by  induction, for the whole graph.  Finally,
every edge is embedded in a page in  constant time, which altogether
takes linear time, since there are at most $4n-8$ edges. Hence,
the algorithm runs in linear time.
\end{proof}

We now consider the special case of a 1-planar graph with a Hamiltonian cycle.
Recall that planar graphs  with a Hamiltionian cycle  can be  embedded
in two pages \cite{bk-btg-79}.

\begin{corollary} (Alam~et al.~\cite{abk-bt1pg-15}) \label{cor:Hamiltonian}
A 1-planar graph $G=(V,E)$ has a 4-page book embedding if the planar
skeleton has a Hamiltonian cycle.
\end{corollary}

\begin{proof}
Suppose the planar skeleton $G[E_b]$
has a Hamiltonian cycle
$\Gamma$. For any pair of crossed edges, arbitrarily color green one edge   and
red the other. If $E_g$ ($E_r$) is the set of green
(red) edges, then both subgraphs $G_1 = G[E_b \cup E_g]$ and $G_2 = [E_b \cup
E_r]$ are planar and contain the  Hamiltonian cycle $\Gamma$.
Using the linear ordering from  $\Gamma$,  embed each of $G_1$
and  $G_2$ in  two pages, respectively, such that $G$ has a 4-page book embedding.
\end{proof}

An $n$-vertex 1-planar graph is called \emph{optimal} if it has $4n-8$
edges, which is maximum for simple 1-planar graphs \cite{bsw-1og-84}.
Extended wheel graphs are important optimal 1-planar graphs
\cite{b-ro1plt-18,s-s1pg-86,s-rm1pg-10}.
An \emph{extended wheel graph} $XW_{2k}$ is a 1-planar graph with $n=2k+2$
vertices and $4n-8$ edges. It consists of a cycle $C  =
 v_1,\ldots v_{2k}$   and
two poles $p$ and $q$. The edges of $C$ are uncrossed. There is an
edge $\edge{p}{v_i}$ and $\edge{q}{v_i}$ for $i=1,\ldots,2k$, but no
edge $\edge{p}{q}$. In addition, there are edges $\edge{v_i}{
v_{i+2}}$ for $i=1,\ldots, 2k$ where $v_{2k+1}=v_1$ and
$v_{2k+2}=v_2$.
The planar skeleton has a Hamiltonian cycle $(p, v_2, v_1, q, v_3, v_4, \ldots, v_{2k}, p)$.

\begin{lemma} \label{lem:optimal}
The book thickness of optimal 1-planar graphs is at least four.
\end{lemma}
\begin{proof}
The book thickness of a graph with $n$ vertices and $m$ edges is at
least $\lceil \frac{m-n}{n-3} \rceil$  \cite{bk-btg-79}.
Optimal 1-planar  $n$-vertex graphs have
$4n-8$ edges, such that $\frac{m-n}{n-3} > 3$.
\end{proof}

\begin{corollary} \label{cor:XWgraphs}
The optimal 1-planar graphs $XW_{2k}$ have  book thickness  four.
\end{corollary}

In consequence, the crossed cube $XW_6$, shown in
Figure~\ref{fig:crossed-cube}, has book thickness four.  Note that also the
  unique optimal 1-planar graph  with 11 vertices
\cite{bsw-1og-84, s-o1pgts-10} can be embedded in 4 pages, although
  its planar skeleton does not have a Hamiltonian cycle.

\section{Conclusion}  \label{sec:conclusion}

We extend Yannakakis   algorithm \cite{y-epg4p-89}  on the book
embedding of planar graphs to 1-planar graphs and show that they can be
embedded in ten pages, which improves all previous bounds.  There are
 small 1-planar graphs with book thickness four, such that
there remains a gap of six between the upper and lower bounds on
the book thickness of 1-planar graphs.
 We conjecture that optimal 1-planar graphs
have   book thickness four.

There are many classes of beyond planar graphs
 for which the book thickness has not yet
been investigated \cite{dlm-survey-beyond-19}, including subclasses of 1-planar graphs,
such as IC- and NIC-planar graphs
\cite{klm-bib-17}.
 It is known that outer 1-planar graphs have book
thickness   two \cite{abbghnr-o1p-16} and
clique-augmented 2-planar graphs have book thickness at most 17
\cite{b-bookmap-20}.



\end{document}